\documentclass{article}
\usepackage{fullpage}
\usepackage{amsthm,amsmath,amssymb,mathtools}
\usepackage{enumitem}
\usepackage[colorlinks,citecolor=blue,bookmarks=true,linktocpage]{hyperref}
\usepackage{aliascnt}
\usepackage[numbered]{bookmark}
\usepackage[capitalise]{cleveref}
\usepackage[backend=biber,url=false,arxiv=abs,maxbibnames=100,isbn=false,sortcites=true]{biblatex}
\usepackage{xcolor}
\usepackage{lineno}
\usepackage{multirow,multicol}
\usepackage{tikz}
\usetikzlibrary{calc,positioning,shapes}

\newcommand{\List}{\mathcal{L}_\mathcal{H}}
\newcommand{\Floor}[1]{{\left\lfloor #1 \right\rfloor}}
\newcommand{\bandw}[1]{\textsf{bw}(#1)}

\newcommand{\Econf}[1]{I^{\rm E}_{#1}}
\newcommand{\Vconf}[1]{I^{\rm V}_{#1}}
\newcommand{\tw}{{\rm tw}}
\newcommand{\pw}{{\rm pw}}
\newcommand{\bw}{{\rm bw}}
\newcommand{\opt}{{\rm opt}}

\newcommand{\intsec}[1]{[#1]}
\newcommand{\intsecz}[1]{[#1] \cup \{0\}}
\newcommand{\landnum}[2]{\lambda(#1,#2)}
\newcommand{\sericom}{\bullet}
\newcommand{\paracom}{\parallel}
\newcommand{\seribool}[2]{b_{#1}^{\sericom}(#2)}
\newcommand{\parabool}[2]{b_{#1}^{\paracom}(#2)}
\newcommand{\seridp}[1]{f^{\sericom}(#1)}
\newcommand{\paradp}[1]{f^{\paracom}(#1)}
\newcommand{\bicon}{\beta}
\newcommand{\incedge}{E}

\newtheorem{theorem}{Theorem}
\newtheorem{lemma}{Lemma}
\newtheorem{proposition}{Proposition}
\newtheorem{corollary}{Corollary}

\title{On the complexity of list $\mathcal H$-packing for sparse graph classes}
\author{
Tatsuya Gima\thanks{Graduate School of Informatics, Nagoya University, Nagoya, Japan and JSPS Research Fellow.
\texttt{gima@nagoya-u.jp}}\and 
Tesshu Hanaka\thanks{Department of Informatics, Kyushu University, Fukuoka, Japan.
\texttt{hanaka@inf.kyushu-u.ac.jp}} \and 
Yasuaki Kobayashi\thanks{Graduate School of Information Science and Technology, Hokkaido University, Hokkaido, Japan.
\texttt{koba@ist.hokudai.ac.jp}} \and 
Yota Otachi\thanks{Graduate School of Informatics, Nagoya University, Nagoya, Japan.
\texttt{otachi@nagoya-u.jp}} \and 
Tomohito Shirai\thanks{Graduate School of Information Sciences, Tohoku University, Sendai, Japan.
\texttt{\{akira,tamura,zhou\}@tohoku.ac.jp}} \and 
Akira Suzuki\footnotemark[5] \and 
Yuma Tamura\footnotemark[5] \and 
Xiao Zhou\footnotemark[5]}

\begin{document}

\maketitle

\begin{abstract}
The problem of packing as many subgraphs isomorphic to $H \in \mathcal H$ as possible in a graph for a class $\mathcal H$ of graphs is well studied in the literature.
Both vertex-disjoint and edge-disjoint versions are known to be NP-complete for $H$ that contains at least three vertices and at least three edges, respectively.
In this paper, we consider ``list variants'' of these problems: Given a graph $G$, an integer $k$, and a collection $\mathcal L_{\mathcal H}$ of subgraphs of $G$ isomorphic to some $H \in \mathcal H$, the goal is to compute $k$ subgraphs in $\mathcal L_{\mathcal H}$ that are pairwise vertex- or edge-disjoint.
We show several positive and negative results, focusing on classes of sparse graphs, such as bounded-degree graphs, planar graphs, and bounded-treewidth graphs.
\end{abstract}

\section{Introduction}

Packing as many graphs as possible into another graph is a fundamental problem in the field of graph algorithms.
To be precise, for a fixed graph $H$, given an undirected graph $G$ and a non-negative integer $k$, the \textsc{Vertex Disjoint $H$-Packing} problem (resp.\ the \textsc{Edge Disjoint $H$-Packing} problem) asks for finding a collection $\mathcal{S}$ of $k$ vertex-disjoint (resp. edge-disjoint) subgraphs of $G$ that are isomorphic to $H$.
For a connected graph $H$, both problems are polynomially solvable if $H$ has at most two vertices (resp.\ at most two edges) because they can be reduced to the \textsc{Maximum Matching} problem, whereas the problems are shown to be NP-complete if $H$ has at least three vertices (resp.\ at least three edges)~\cite{CorneilMH94:DAM:Edge-disjoint,KirkpatrickH83:SICOMP:complexity}.
Furthermore, both problems are naturally extended to \textsc{Vertex Disjoint $\mathcal{H}$-Packing} and \textsc{Edge Disjoint $\mathcal{H}$-Packing}~\cite{ChungG81,KirkpatrickH83:SICOMP:complexity}, which respectively ask for finding a collection $\mathcal{S}$ of $k$ vertex-disjoint and edge-disjoint subgraphs of $G$ that are isomorphic to some graph in a (possibly infinite) fixed collection $\mathcal{H}$ of graphs.
These problems are also well studied in specific cases of $\mathcal{H}$.
In particular, when $\mathcal{H}$ is paths or cycles, it has received much attention in the literature because of the variety of possible applications~\cite{BafnaP93,BermanJLSS90,BontridderHHHLRS03,KosowskiMZ05,MalafiejskiZ05}.
In both cases, \textsc{Vertex Disjoint $\mathcal{H}$-Packing} and \textsc{Edge Disjoint $\mathcal{H}$-Packing} remain NP-complete for planar graphs~\cite{BermanJLSS90,DyerF85}.

Recently, Xu and Zhang proposed a new variant of \textsc{Edge Disjoint $\mathcal{H}$-Packing}, which they call \textsc{Path Set Packing}, from the perspective of network design~\cite{XuZ18}.
In the \textsc{Path Set Packing} problem, given an undirected graph $G$, a non-negative integer $k$, and a collection $\mathcal{L}$ of simple paths in $G$, we are required to find a subcollection $\mathcal{S} \subseteq \mathcal{L}$ of (at least) $k$ paths that are mutually edge-disjoint. 
Notice that $\mathcal{L}$ may not be exhaustive: Some paths in $G$ may not appear in $\mathcal L$.
If $\mathcal H$ consists of a finite number of paths, \textsc{Edge Disjoint $\mathcal{H}$-Packing} can be (polynomially) reduced to \textsc{Path Set Packing} because $\mathcal H$ is fixed and hence all paths in $G$ isomorphic to some graph in $\mathcal{H}$ can be enumerated in polynomial time.
Xu and Zhang showed that for a graph $G$ with $n$ vertices and $m$ edges, the optimization variant of \textsc{Path Set Packing} is hard to approximate within a factor $O(m^{1/2-\epsilon})$ for any constant $\epsilon > 0$ unless $\text{NP $=$ ZPP}$, while the problem is solvable in $O(|\mathcal{L}|n^2)$ time if $G$ is a tree and in $O(|\mathcal{L}|^{\tw \Delta} n)$ time if $G$ has treewidth $\tw$ and maximum degree $\Delta$~\cite{XuZ18}.
Very recently, Aravind and Saxena investigated the parameterized complexity of \textsc{Path Set Packing} for various parameters.
For instance, \textsc{Path Set Packing} is W[1]-hard even when parameterized by pathwidth plus maximum degree plus solution size~\cite{AravindS23}.
To the best of our knowledge, except for \textsc{Path Set Packing}, such a variant has not been studied for \textsc{Edge Disjoint $\mathcal{H}$-Packing}, nor \textsc{Vertex Disjoint $\mathcal{H}$-Packing}.

\paragraph{Our contributions.}
In this paper, motivated by \textsc{Path Set Packing}, we introduce \emph{list variants} of \textsc{Vertex Disjoint $\mathcal{H}$-Packing} and \textsc{Edge Disjoint $\mathcal{H}$-Packing}.
In the \textsc{Vertex Disjoint List $\mathcal{H}$-Packing} (resp.\ \textsc{Edge Disjoint List $\mathcal{H}$-Packing}) problem, we are given a graph $G$, a non-negative integer $k$, and a collection (list) $\List$ of subgraphs of $G$ such that each subgraph in $\List$ is isomorphic to some graph in $\mathcal{H}$. 
The problem asks whether there exists a subcollection $\mathcal{S} \subseteq \List$ such that $|\mathcal{S}| \ge k$ and subgraphs of $G$ in $\mathcal{S}$ are vertex-disjoint (resp.\ edge-disjoint). 
If $\List$ contains all subgraphs of $G$ isomorphic to some graph in $\mathcal{H}$, the problem is equivalent to \textsc{Vertex Disjoint $\mathcal{H}$-Packing} (resp.\ \textsc{Edge Disjoint $\mathcal{H}$-Packing}).
Thus, the tractability of the list variants implies that of the original problems.
If $\mathcal{H} = \{ H \}$ for a fixed graph $H$, then we call these problems simply \textsc{Vertex (Edge) Disjoint List $H$-Packing} and \textsc{Vertex (Edge) Disjoint $H$-Packing}, respectively.
For a positive integer $\ell$, we denote by $P_\ell$ and $C_\ell$ the path and the cycle of $\ell$ vertices, respectively.
(We assume $\ell \ge 3$ for $C_\ell$.)
When $\mathcal{P} = \{P_\ell : \ell \ge 1 \}$, \textsc{Edge Disjoint List $\mathcal{P}$-Packing} is equivalent to \textsc{Path Set Packing}.
Therefore, \textsc{Edge Disjoint List $\mathcal{H}$-Packing} generalizes both \textsc{Edge Disjoint $H$-Packing} and \textsc{Path Set Packing}.

We first give sufficient conditions to solve \textsc{Vertex (Edge) Disjoint List $H$-Packing} on bounded degree graphs in polynomial time.
These conditions directly indicate the polynomial-time solvability of \textsc{Vertex (Edge) Disjoint List $C_\ell$-Packing} on graphs of maximum degree $3$ for $\ell \in \{ 3,4,5 \}$.
It is worth noting that \textsc{Vertex Disjoint $P_3$-Packing} remains NP-complete even for 2-connected bipartite planar cubic graphs~\cite{MalafiejskiZ05}.
In contrast, we show that \textsc{Vertex (Edge) Disjoint $C_\ell$-Packing} on planar graphs of maximum degree $3$ is NP-complete for any $\ell \ge 6$.
As \textsc{Vertex (Edge) Disjoint $C_\ell$-Packing} can be represented by its list variant, this result also indicates the hardness of \textsc{Vertex (Edge) Disjoint List $C_\ell$-Packing}.
We also give the NP-completeness of \textsc{Vertex (Edge) Disjoint $C_\ell$-Packing} on planar graphs of maximum degree $4$ for any $\ell \ge 4$.
Therefore, we provide the complexity dichotomy of \textsc{Vertex (Edge) Disjoint $C_\ell$-Packing} with respect to the maximum degree of a given graph and $\ell$, as summarized in \tableautorefname~\ref{tbl:C_l}.

\begin{table}[t]
    \caption{The complexity of \textsc{Vertex Disjoint $C_{\ell}$-Packing} (left) and \textsc{Edge Disjoint $C_{\ell}$-Packing} (right).}\label{tbl:C_l}
    
    \begin{minipage}[t]{.5\textwidth}
    \centering
    \begin{tabular}{|c|c|c|}
    \hline
          & $\Delta(G) = 3$    & $\Delta(G) \ge 4$    \\ \hline
    $\ell = 3$ & P \cite{CapraraR02:IPL:Packing} & NPC \cite{CapraraR02:IPL:Packing}  \\ \hline
    $\ell = 4$ & \multirow{2}{*}{P [Thm \ref{thm:vertex-edge:C_4:deg3}]} & \multirow{2}{*}{NPC [Thm \ref{thm:vertex:C_4:deg4}]} \\ \cline{1-1}
    $\ell = 5$ & & \\ \cline{1-3}
    $\ell \ge 6$ & \multicolumn{2}{c|}{NPC [Thm \ref{thm:vertex-edge:C_6:deg3}]} \\ \hline
    \end{tabular}
    \end{minipage}
    \hfill
    \begin{minipage}[t]{.5\textwidth}
    \begin{tabular}{|c|c|c|c|}
    \hline
     & $\Delta(G) = 3$ & \multicolumn{1}{l|}{$\Delta(G) = 4$} & $\Delta(G) \ge 5$ \\ \hline
     $\ell = 3$ & \multicolumn{2}{c|}{P \cite{CapraraR02:IPL:Packing}} & NPC \cite{CapraraR02:IPL:Packing} \\ \hline
     $\ell = 4$ &\multirow{2}{*}{P [Thm \ref{thm:vertex-edge:C_4:deg3}]}&
     \multicolumn{2}{c|}{\multirow{2}{*}{NPC [Thm \ref{thm:edge:C_4:deg4}]}}\\ \cline{1-1}
     $\ell = 5$ & & \multicolumn{2}{l|}{} \\ \cline{1-4}
     $\ell \ge 6$ & \multicolumn{3}{c|}{NPC [Thm \ref{thm:vertex-edge:C_6:deg3}]} \\ \hline
    \end{tabular}
    \end{minipage}
\end{table}

Second, we design a polynomial-time algorithm for \textsc{Vertex Disjoint List $\mathcal{H}$-Packing} on bounded-treewidth graphs, provided that all graphs in $\mathcal{H}$ are connected.
This implies that \textsc{Vertex Disjoint List $\mathcal{H}$-Packing} belongs to XP parameterized by treewidth.
Note that the connectivity condition on $\mathcal H$ is essential, because otherwise one can see that the problem is NP-complete even on forests (see~\Cref{thm:H-packing-exception}).
On the other hand, we show that \textsc{Vertex Disjoint List $\mathcal{P}$-Packing} and \textsc{Vertex Disjoint List $\mathcal{C}$-Packing} parameterized by pathwidth plus $k$ are W[1]-hard, where $\mathcal{P} = \{P_\ell : \ell \ge 1 \}$ and $\mathcal{C} = \{C_\ell : \ell \ge 3\}$.
This result implies that there is probably no FPT algorithm for the problems parameterized by treewidth.
One might think that XP algorithms parameterized by treewidth could also be designed for the edge-disjoint versions.
We give a negative answer.
We show that \textsc{Edge Disjoint List $\mathcal{P}$-Packing} and \textsc{Edge Disjoint List $\mathcal{C}$-Packing} parameterized by bandwidth plus $k$ are W[1]-hard even for outerplanar and two-terminal series-parallel graphs, which have treewidth at most~$2$.
In particular, the W[1]-hardness for \textsc{Edge Disjoint List $\mathcal{P}$-Packing}, which is equivalent to \textsc{Path Set Packing}, strengthens the result of~\cite{AravindS23}.

The above hardness results prompt us to further investigate the complexity of \textsc{Edge Disjoint List $P_\ell$-Packing} and \textsc{Edge Disjoint List $C_\ell$-Packing} on bounded-treewidth graphs.
In this paper, we focus on series-parallel graphs, also known as graphs of treewidth at most~$2$.
We show that \textsc{Edge Disjoint List $P_4$-Packing} and \textsc{Edge Disjoint List $C_5$-Packing} remain NP-complete even for series-parallel graphs.
Since \textsc{Edge Disjoint List $P_3$-Packing} is solvable in polynomial time for general graphs by reducing to \textsc{Maximum Matching}, the former implies the complexity dichotomy of \textsc{Edge Disjoint List $P_\ell$-Packing} on series-parallel graphs with respect to $\ell$.
The remaining task is to settle the complexity of \textsc{Edge Disjoint List $C_\ell$-Packing} on series-parallel graphs for $\ell \le 4$.
We finally provide an algorithm that, given an $n$-vertex series-parallel graph and a collection $\List$ of cycles with length $\ell \le 4$, solves \textsc{Edge Disjoint List $C_\ell$-Packing} in $O(|\List| + n^{2.5})$ time.


	




\paragraph{Related work.}
The packing problems have a long history.
\textsc{Vertex Disjoint $\mathcal{C}$-Packing} (also known simply as \textsc{Cycle Packing}) has been studied in numerous papers after the well-known Erd\H{o}s-P\'{o}sa theorem was given~\cite{ErdosP65}.
In particular, \textsc{Cycle Packing} has received a lot of attention in the field of parameterized complexity.
\textsc{Cycle Packing} is fixed-parameter tractable when parameterized by a solution size $k$~\cite{Bodlaender94,LokshtanovMSZ19} or by treewidth~\cite{CyganFKLMPPS15}.
The complexity of restricted variants of \textsc{Cycle Packing} has also been explored. 
For any $\ell \ge 3$, \textsc{Vertex Disjoint $C_\ell$-Packing} on planar graphs is NP-complete~\cite{BermanJLSS90}.
In addition, \textsc{Vertex Disjoint $C_3$-Packing} is NP-complete for planar graphs of maximum degree $4$~\cite{CapraraR02:IPL:Packing}, chordal graphs, line graphs, and total graphs~\cite{GuruswamiRCCW98}, while the problem is solvable in polynomial time for graphs of maximum degree $3$~\cite{CapraraR02:IPL:Packing}, split graphs, and cographs~\cite{GuruswamiRCCW98}.

In the edge-disjoint variant of \textsc{Cycle Packing}, for any $\ell \ge 3$, \textsc{Edge Disjoint $C_\ell$-Packing} is NP-complete for planar graphs~\cite{HeathV98} and balanced $2$-interval graphs~\cite{JiangXZ16}.
Furthermore, for any even $\ell \ge 4$, \textsc{Edge Disjoint $C_\ell$-Packing} on bipartite graphs and balanced $2$-track interval graphs is NP-complete~\cite{JiangXZ16}.
\textsc{Edge Disjoint $C_3$-Packing} is solvable for graphs of maximum degree $4$~\cite{CapraraR02:IPL:Packing} and outerplanar graphs~\cite{HeathV98}, while the problem is NP-complete for planar graphs of maximum degree $5$~\cite{CapraraR02:IPL:Packing}.

As for problems of packing disjoint paths, \textsc{Vertex Disjoint $P_\ell$-Packing} for any $\ell \ge 3$ is NP-complete for graphs of maximum degree $3$ and \textsc{Edge Disjoint $P_\ell$-Packing} for any $\ell \ge 4$ is NP-complete for graphs of maximum degree $4$~\cite{MasuyamaI91}.
Moreover, \textsc{Vertex Disjoint $P_3$-Packing} remains NP-complete even for 2-connected bipartite planar cubic graphs~\cite{MalafiejskiZ05}.
On the other hand, for any positive integer $\ell$, \textsc{Vertex Disjoint $P_\ell$-Packing} is solvable in polynomial time for trees~\cite{MasuyamaI91}.
\textsc{Vertex Disjoint $P_2$-Packing} is equivalent to \textsc{Maximum Matching}, which is solvable in polynomial time~\cite{HopcroftT73}.
\textsc{Vertex Disjoint $P_3$-Packing} is polynomially solvable for certain cases~\cite{AkiyamaC90,KanekoKN01}, and the problem is in FPT and has a linear kernel when parameterized by a solution size~\cite{FernauR09,PrietoS06,WangNFC08}.

Compared to the other packing problems, as far as we know, not much is known about tractable cases of \textsc{Edge Disjoint $P_\ell$-Packing}.
\textsc{Edge Disjoint $P_\ell$-Packing} for any positive integer $\ell$ is solvable in polynomial time for trees~\cite{MasuyamaI91} and \textsc{Edge Disjoint $P_3$-Packing} is solvable in linear time for general graphs~\cite{MasuyamaI91}.

\textsc{Vertex (Edge) Disjoint $\mathcal{H}$-Packing} is closely related to the \textsc{Set Packing} problem.
Given a universe $U$, a non-negative integer $k$, and a collection $\mathcal{L}$ of subsets of $U$, \textsc{Set Packing} asks whether there exists a subcollection $\mathcal{S} \subseteq \mathcal{L}$ such that $|\mathcal{S}| \ge k$ and $\mathcal{S}$ are mutually disjoint.
If $\mathcal{H}$ contains all (possibly disconnected) graphs, then these problems are equivalent because it is possible to construct a graph whose vertices (edges) correspond to 
 elements of $U$.
If $\mathcal{H}$ contains restricted graphs, then \textsc{Vertex (Edge) Disjoint $\mathcal{H}$-Packing} are special cases of \textsc{Set Packing}.
Therefore, all tractable results for \textsc{Set Packing} are applicable to our problems.
This paper further seeks tractable cases of \textsc{Set Packing} from the viewpoint of graph structure.



\section{Preliminaries}
For a positive integer $i$, we denote $\intsec{i} = \{1,2,\ldots, i\}$.

Let $G$ be a graph.
Throughout this paper, we assume that $G$ is simple, that is, it has neither self-loops nor parallel edges.
The sets of vertices and edges of $G$ are denoted by $V(G)$ and $E(G)$, respectively.
For $v \in V(G)$, we denote by $N_G(v)$ the set of neighbor of $v$, by $\incedge_G(v)$ the set of incident edges of $v$, and by $d_G(v)$ the degree of $v$ in $G$.
The maximum degree of a vertex in $G$ is denoted by $\Delta(G)$ and the minimum degree of a vertex in $G$ is denoted by $\delta(G)$.
We may simply write $uv$ to denote an edge $\{u, v\}$.
For a positive integer $t$, we denote by $tG$ the disjoint union of $t$ copies of $G$.
For a graph $H$, the \emph{$H$-vertex-conflict graph} of $G$, denoted $\Vconf{H}(G)$, is defined as follows.
Each vertex of $\Vconf{H}(G)$ corresponds to a subgraph isomorphic to $H$ in $G$.
Two vertices of $\Vconf{H}(G)$ are adjacent if and only if the corresponding subgraphs in $G$ share a vertex.
The \emph{$H$-edge-conflict graph} of $G$, denoted $\Econf{H}(G)$, is defined by replacing the adjacency condition in the definition of $\Vconf{H}(G)$ as: two vertices of $\Econf{H}(G)$ are adjacent if and only if they share an edge in $G$.

A \emph{claw} is a star graph with three leaves.
A graph is said to be \emph{claw-free} if it has no claw as an induced subgraph.
Minty~\cite{Minty80:JCTB:maximal} and Sbihi~\cite{Sbihi80:DM:Algorithme} showed that the maximum independent set problem can be solved in polynomial time on claw-free graphs.
This immediately implies the following proposition, which is a key to our polynomial-time algorithms.

\begin{proposition}\label{prop:claw-free}
    If $\Vconf{H}(G)$ is claw-free, then \textsc{Vertex Disjoint List $H$-Packing} can be solved in $n^{O(|V(H)|)}$ time.
    Moreover, if $\Econf{H}(G)$ is claw-free, then \textsc{Edge Disjoint List $H$-Packing} can be solved in $n^{O(|V(H)|)}$ time as well.
\end{proposition}

A \emph{tree decomposition} of $G$ is a pair $\mathcal T = (T, \{X_i : i \in V(T)\})$ of a tree $T$ and vertex sets $\{X_i : i \in V(T)\}$, called \emph{bags}, that satisfies the following conditions: $\bigcup_{i \in V(T)} X_i = V(G)$; for $\{u, v\} \in E(G)$, there is $i \in V(T)$ such that $\{u, v\} \subseteq X_i$; and for $v \in V(T)$, the bags containing $v$ induces a subtree of $T$.
The \emph{width} of $\mathcal T$ is defined as $\max_{i \in V(T)}|X_i| - 1$, and the \emph{treewidth} of $G$ is the minimum integer $k$ such that $G$ has a tree decomposition of width $k$.
A \emph{path decomposition} of $G$ is a tree decomposition of $G$, where $T$ forms a path.
The \emph{pathwidth} of $G$ is defined analogously.
Let $\pi \colon V(G) \to [|V(G)|]$ be a bijection, which we call a \emph{linear layout} of $G$.
The width of a linear layout $\pi$ of $G$ is defined as $\max_{\{u, v\} \in E(G)}|\pi(u) - \pi(v)|$.
The \emph{bandwidth} of $G$ is the minimum integer $k$ such that $G$ has a linear layout of width $k$.

For the treewidth, pathwidth, and bandwidth of $G$, we denote them by $\tw(G)$, $\pw(G)$, and $\bw(G)$.
We may simply write them as $\tw$, $\pw$, and $\bw$ without specific reference to $G$.
It is well known that $\tw(G) \le \pw(G) \le \bw(G)$ for every graph $G$~\cite{Bodlaender98}.

\section{List \texorpdfstring{$C_\ell$}{}-packing on bounded degree graphs}\label{sec:bnddeg}
In this section, we focus on \textsc{Vertex Disjoint List $C_\ell$-Packing} and \textsc{Edge Disjoint List $C_\ell$-Packing} on bounded degree graphs.
We give sufficient conditions to solve \textsc{Vertex (Edge) Disjoint $H$-Packing} in polynomial time, which extend polynomial-time algorithms for $\ell = 3$~\cite{CapraraR02:IPL:Packing} to more general cases.
We also show that these conditions are ``tight'' in some sense.

\begin{theorem}\label{thm:vertex:H-packing}
    \textsc{Vertex Disjoint List $H$-Packing} can be solved in polynomial time if the following inequality holds:
    \begin{align*}
        \Delta(G) \le 2\delta(H) - \Floor{\frac{|V(H)|}{3}}.
    \end{align*}
\end{theorem}
\begin{proof}
    By~\Cref{prop:claw-free}, it suffices to show that if $\Vconf{H}(G)$ has a claw as an induced subgraph, then $G$ has a vertex of degree more than $2\delta(H) - \Floor{|V(H)|/3}$.
    Let $H^*$, $H_1$, $H_2$, and $H_3$ be induced copies of $H$ in $G$ such that they correspond to an induced claw in $\Vconf{H}(G)$ whose center is $H^*$.
    This implies that $V(H^*) \cap V(H_i) \neq \emptyset$ for $1\le i \le 3$ and $V(H_i) \cap V(H_{j}) = \emptyset$ for $1 \le i < j \le 3$.
    This implies that some copy of $H$, say $H_1$, satisfies $|V(H^*) \cap V(H_1)| \le \Floor{|V(H)|/3}$.
    Let $v \in V(H^*) \cap V(H_1)$.
    The vertex $v$ has at least $\delta(H)$ neighbors in each of $H^*$ and $H_1$ and at most $\Floor{|V(H)|/3} - 1$ of them belong to $V(H^*) \cap V(H_1)$.
    Thus,
    \begin{align*}
        d_G(v) &\ge 2\delta(H) - (|V(H^*)\cap V(H_1)| - 1) \\
        &\ge 2\delta(H) - \left(\Floor{\frac{|V(H)|}{3}} - 1\right) \\
        &> 2\delta(H) - \Floor{\frac{|V(H)|}{3}},
    \end{align*}
    which proves the claim.
\end{proof}

\begin{theorem}\label{thm:edge:H-packing}
    \textsc{Edge Disjoint List $H$-Packing} can be solved in polynomial time if the following inequality holds:
    \begin{align*}
        \Delta(G) \le 2\delta(H) - \Floor{\frac{|E(H)|}{3}},
    \end{align*}
    except that $H = tK_2$ for $3 \le t \le 5$.
\end{theorem}
\begin{proof}
    We first assume that $|E(H)| \ge 6$.
    The proof strategy is analogous to that in \Cref{thm:vertex:H-packing}.
    Let $H^*$, $H_1$, $H_2$, and $H_3$ be induced copies of $H$ in $G$ such that they correspond to an induced claw in $\Vconf{H}(G)$ whose center is $H^*$.
    We assume that $|E(H^*) \cap E(H_1)| \le \Floor{|E(H)| / 3}$.
    As $E(H^*) \cap E(H_1) \neq \emptyset$, there are at least two vertices, say $v$ and $w$, in $V(H^*) \cap V(H_1)$.
    Both $v$ and $w$ have at least $\delta(H)$ incident edges in each of $H^*$ and $H_1$.
    If $v$ has at most $\Floor{|E(H)|/3} - 1$ incident edges that belong to $E(H^*) \cap E(H_1)$, then we have
    \begin{align*}
        d_G(v) \ge 2\delta(H) - \left(\Floor{\frac{|E(H)|}{3}} - 1 \right) > 2\delta(H) - \Floor{\frac{|E(H)|}{3}}.
    \end{align*}
    Suppose otherwise. Then, all the edges in $E(H^*) \cap E(H_1)$ are incident to $v$.
    This implies that, at most one edge (i.e., $\{v, w\}$ if it exists) in $E(H^*) \cap E(H_1)$ can be incident to $w$.
    Thus, we have $d_G(w) \ge 2\delta(H) - 1$.
    By the assumption $|E(H)| \ge 6$, we have $d_G(w) > 2\delta(H) - \Floor{|E(H)|/3}$.

    We next consider the other case, that is, $|E(H)| \le 5$, and show that \textsc{Edge Disjoint List $H$-Packing} is solvable in polynomial time under the assumption that $\Delta(G) \le 2\delta(H) - \Floor{|E(H)|/3}$.
    We can assume that $|E(H)| \ge 3$ as otherwise $\Econf{H}(G)$ has no induced claws.
    If $\delta(H) \ge 3$, then it has at least four vertices, and thus we have $|E(H)|\ge |V(H)| \cdot\delta(H) / 2 \ge 4\delta(H)/2 \ge 6$.
    Thus, we assume $1 \le \delta(H) \le 2$.
    If $\delta(H) = 2$, then $\Delta(G) \le 4 - \Floor{|E(H)|/3} = 3$.
    In this case, two copies of $H$ appearing in $G$ are edge-disjoint if and only if they are vertex-disjoint.
    Moreover, as $|E(H)| \ge |V(H)|$, we have $\Delta(G) \le 2\delta(H) - \Floor{|V(H)|/3}$, yields a polynomial-time algorithm by~\Cref{thm:vertex:H-packing}.
    Thus, the remaining case is $\delta(H) = 1$, that is, $H = tK_2$ for $3 \le t \le 5$.
\end{proof}

Let us note that the exception in \Cref{thm:edge:H-packing} is critical for the tractability of \textsc{Edge Disjoint List $H$-Packing}.
In fact, \textsc{Edge Disjoint List $H$-Packing} is NP-complete even if $G = nK_2$ and $H = 3K_2$, which satisfy that $\Delta(G) = 1 \le 2\delta(H) - \Floor{|E(H)|/3}$.
This intractability is shown by reducing \textsc{Exact Cover by $3$-Sets}, which is known to be NP-complete~\cite{Karp72}, to \textsc{Edge Disjoint List $H$-Packing}.
In \textsc{Exact Cover by $3$-Sets}, we are given a universe $U$ and a collection $\mathcal C$ of subsets of $U$, each of which has exactly three elements and asked whether there is a pairwise disjoint subcollection $\mathcal C'$ of $\mathcal C$ that covers $U$ (i.e., $U = \bigcup_{S \in \mathcal C'} S$).
This problem is a special case of \textsc{Edge Disjoint List $H$-Packing}, where $G$ has a copy of $K_2$ corresponding to each element in $U$ and $\mathcal S$ consists of all copies of $3K_2$ corresponding to $\mathcal C$.
This reduction also proves that \textsc{Vertex Disjoint List $H$-Packing} is NP-complete even if $G = nK_2$ and $H = 3K_2$.

\begin{theorem}\label{thm:H-packing-exception}
\textsc{Vertex Disjoint List $H$-Packing} and \textsc{Edge Disjoint List $H$-Packing} are NP-complete even if $G = nK_2$ and $H = 3K_2$.
\end{theorem}

As consequences of \Cref{thm:vertex:H-packing,thm:edge:H-packing}, we have the following positive results.
\begin{theorem}\label{thm:vertex-edge:C_4:deg3}
    For $\ell \in \{4, 5\}$, there are polynomial-time algorithms for \textsc{Vertex Disjoint List $C_\ell$-Packing} and \textsc{Edge Disjoint List $C_\ell$-Packing} on graphs of maximum degree $3$.
\end{theorem}

Contrary to this tractability, for any $\ell \ge 6$, \textsc{Vertex Disjoint $C_\ell$-Packing} and \textsc{Edge Disjoint $C_\ell$-Packing} are NP-complete even on planar graphs of maximum degree $3$.

\begin{theorem}\label{thm:vertex-edge:C_6:deg3}
    For $\ell \ge 6$, \textsc{Vertex Disjoint $C_\ell$-Packing} and \textsc{Edge Disjoint $C_\ell$-Packing} are NP-complete even on planar graphs of maximum degree $3$.
\end{theorem}
\begin{proof}
    Since \textsc{Vertex Disjoint $C_\ell$-Packing} and \textsc{Edge Disjoint $C_\ell$-Packing} are equivalent on graphs of maximum degree $3$, we only consider \textsc{Vertex Disjoint $C_\ell$-Packing}.

    To show NP-hardness, we perform a polynomial-time reduction from \textsc{Independent Set}, which is known to be NP-complete even on planar graphs with maximum degree $3$ and girth at least $p$ for any integer $p$~\cite{Murphy92:DAM:Computing}.
    Here, the \emph{girth} of $G$ is the length of a shortest cycle in $G$.

    Let $G$ be a planar graph with $\Delta(G) \le 3$ and girth at least $\ell + 1$.
    We construct a graph $G'$ as follows.
    For each $v \in V(G)$, $G'$ contains a cycle $C_v$ of length $\ell$.
    These cycles are called \emph{primal cycles} in $G'$.
    Let $M_v$ be an arbitrary matching in $C_v$ with $|M_v| = 3$.
    For two adjacent vertices $u, v$ in $G$, we identify one of the edges in $M_u$ and in $M_v$ as shown in \Cref{fig:red-cycle-bnd-deg}.
    \begin{figure}[t]
        \centering
        \includegraphics{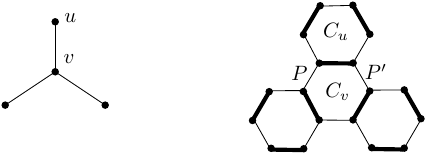}
        \caption{An example of the construction of $G'$ for $\ell = 6$. Bold lines indicate a matching of each cycle.}
        \label{fig:red-cycle-bnd-deg}
    \end{figure}
    The edges in $M_v$ are called \emph{shared edges}.
    By removing the shared edges from $C_v$, we obtain three paths, which we call \emph{private paths} in $C_v$.
    Note that each private path belongs to exactly one primal cycle in $G'$.
    The construction of $G'$ is done.
    It is easy to observe that $G'$ is planar and has maximum degree $3$.
    Also observe that the degree of each internal vertex of a private path is exactly $2$.

    We first claim that the length of cycles in $G'$ except for primal cycles is greater than $\ell$.
    To see this, let $C$ be an arbitrary cycle that is not a primal cycle in $G'$.
    As shared edges form a matching in $G'$, $C$ must contain at least one private path $P$ in $C_v$ for some $v \in V(G)$.
    If $C$ contains all the private paths in $C_v$, the length of $C$ is greater than $\ell$, except for the case $C = C_v$.
    Thus, we assume that $C$ does not contain one of the private paths, say $P'$, in $C_v$. 
    Let $W = (v_0, v_1, \ldots, v_t)$ be a sequence of vertices of $G$ defined as follows.
    We first contract all shared edges in $C$.
    Then, the contracted cycle can be partitioned into maximal subpaths $P_0, \ldots, P_t$ such that each $P_i$ consists of edges of (possibly more than one) private paths in $C_{v_i}$ for some $v_i \in V(G)$.
    Due to the maximality of $P_i$, we have $v_i \neq v_{i+1}$ for $0 \le i \le t$, where the addition in the subscript is taken modulo $t+1$.
    Moreover, the sequence contains at least two vertices as $C$ contains $P$ and does not contain $P'$, meaning that it must have a private path in $C_{v'}$ for some $v' \neq v$.
    For any pair of private paths in $C_u$ and in $C_w$, they are adjacent with a shared edge if and only if $u$ and $w$ are adjacent in $G$.
    Thus, $W$ is a closed walk in $G$.

    Suppose that $W$ contains a ``turn'', that is, $v_i = v_{i + 2}$ for some $i$.
    This implies that $C$ contains all private paths of $C_{v_{i+1}}$, which implies that the length of $C$ is more than $\ell$.
    Otherwise, $W$ contains a cycle in $G$.
    Since the girth of $G$ is at least $\ell + 1$, $W$ has more than $\ell$ edges.
    Hence, $C$ contains more than $\ell + 1$ private paths.

    Now, we are ready to prove that $G$ has an independent set of size at least $k$ if and only if $G'$ has a $C_\ell$-packing of size at least $k$.
    From an independent set of $G$, we can construct a vertex-disjoint $C_\ell$-packing by just taking primal cycles corresponding to vertices in the independent set.
    Since every cycle except for primal cycles has length more than $\ell$, this correspondence is reversible: From a vertex-disjoint $C_\ell$-packing of $G'$ with size $k$, we can construct an independent set of $G$ with size $k$.
\end{proof}

Using a similar strategy of \Cref{thm:vertex-edge:C_6:deg3}, we prove the following theorems.

\begin{theorem}\label{thm:vertex:C_4:deg4}
    For $\ell \in \{4, 5\}$, \textsc{Vertex Disjoint $C_\ell$-Packing} is NP-complete even on planar graphs of maximum degree $4$.
\end{theorem}
\begin{proof}
    We again give a polynomial-time reduction from \textsc{Independent Set} on planar graph with maximum degree at most $3$ and girth at least $\ell + 1$.
    Let $G$ be a planar graph with maximum degree $3$ and girth at least $\ell + 1$.
    We construct a graph $G'$ such that $G$ has an independent set of size $k$ if and only if $G'$ has $k$ vertex-disjoint $C_\ell$'s.
    The construction is almost analogous to one used in \Cref{thm:vertex-edge:C_6:deg3}.
    For each $v \in V(H)$, $G'$ contains a cycle $C_v$ of length $\ell$.
    These cycles are called primal cycles in $G'$.
    We take arbitrary three vertices of $C_v$, which are called \emph{shared vertices}.
    For two adjacent vertices $u, v$ in $G$, we identify one of the shared vertices in $C_u$ and in $C_v$ as shown in \Cref{fig:red-cycle-four-vertex-deg}.
    It is easy to see that the constructed graph $G'$ is planar and has degree at most $4$.
    \begin{figure}[t]
        \centering
        \includegraphics{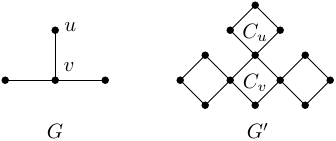}
        \caption{An example of the construction of $G'$ for $\ell = 4$.}
        \label{fig:red-cycle-four-vertex-deg}
    \end{figure}

    Let $C$ be a cycle of $G'$ that is not a primal cycle.
    Since each primal cycle has only three shared vertices, $C$ passes through a primal cycle at most once, where we say that $C$ passes through a primal cycle $C_v$ when $C$ contains at least one edge of $C_v$.
    This fact implies that $C$ induces a cycle of $G$ by tracing $C$ as in the proof of \Cref{thm:vertex-edge:C_6:deg3}.
    Thus, we can conclude that $C$ has more than $\ell$ edges.

    The remaining part of the proof is identical to that in \Cref{thm:vertex-edge:C_6:deg3}.
    For every independent set $S$ of $G$, the primal cycles of them are vertex-disjoint $C_\ell$'s in $G$ and vice versa.
\end{proof}

\begin{theorem}\label{thm:edge:C_4:deg4}
    For $\ell \in \{4, 5\}$, \textsc{Edge Disjoint $C_\ell$-Packing} is NP-complete even on planar graphs of maximum degree $4$.
\end{theorem}
\begin{proof}
    The proof is done by showing a polynomial-time reduction from \textsc{Independent Set} on planar graphs with maximum degree $3$.
    Let $G$ be a planar graph with maximum degree $3$.
    We construct a planar graph $\hat{G}$ with maximum degree at most $4$ such that $G$ has an independent set of size $k$ if and only if $\hat{G}$ has $k'$ edge-disjoint $C_\ell$'s for some $k'$.
    The construction is done in the same spirit of \Cref{thm:vertex-edge:C_6:deg3,thm:vertex:C_4:deg4} but we need to make a minor modification to keep the upper bound on its degree.
    We first subdivide each edge of $G$ twice.
    The subdivided graph is denoted by $G'$.
    It is well known that $G$ has an independent set of size $k$ if and only if $G'$ has an independent set of size $k + |E(G)|$~\cite{Poljak74}.
    For each vertex $v \in V(G')$, $\hat{G}$ contains a cycle of length $\ell$, which we call a primal cycle of $v$.
    For such a cycle $C_v$, we take arbitrary three edges if the degree of $v$ is $3$ in $G'$, and take a matching with two edges otherwise.
    These edges are called shared edges.
    We then identify one of shared edges in $C_u$ and that in $C_v$ if they are adjacent in $G'$.
    See~\Cref{fig:red-cycle-four-deg} for an illustration.
    \begin{figure}[t]
        \centering
        \includegraphics{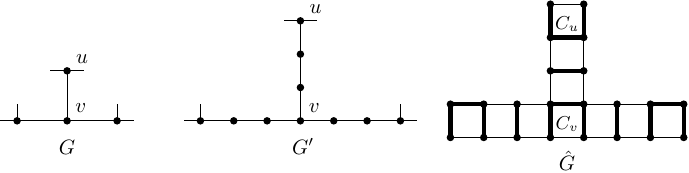}
        \caption{An example of the constructions of $H'$ and $G$ for $\ell = 4$. Bold lines depict shared edges.}
        \label{fig:red-cycle-four-deg}
    \end{figure}
    It is easy to see that $\hat{G}$ is planar.
    As $G'$ has no pair of adjacent vertices of degree $3$, for each vertex in $\hat{G}$, there are at most two shared edges and at most two non-shared edges incident to it, which implies that the maximum degree of $\hat{G}$ is at most $4$.

    Now, we show that $G'$ has an independent set of size $k'$ if and only if $\hat{G}$ has $k'$ edge-disjoint $C_\ell$'s.
    This can be shown by observing that every cycle of $\hat{G}$ has more than $\ell$ edges except for primal cycles.
    The remaining part of the proof is analogous to ones in \Cref{thm:vertex-edge:C_6:deg3,thm:vertex:C_4:deg4}. 
\end{proof}

\section{\textsc{Vertex Disjoint List \texorpdfstring{$\mathcal H$}{}-Packing} on bounded-treewidth graphs}\label{sec:vertex:tw}
In \Cref{sec:series-parallel}, we will see that \textsc{Edge Disjoint List $\mathcal H$-Packing} is intractable even if $\mathcal H$ contains a single small connected graph and an input graph is series-parallel.
In contrast to this intractability, \textsc{Vertex Disjoint List $\mathcal H$-Packing} is polynomial-time solvable on series-parallel graphs and, more generally, bounded-treewidth graphs, if $\mathcal H$ consists of a finite number of connected graphs.
More precisely, we show that \textsc{Vertex Disjoint List $\mathcal H$-Packing} is XP parameterized by treewidth, provided that $\mathcal H$ consists of connected graphs.
We also show that \textsc{Vertex Disjoint List $\mathcal{P}$-Packing} and \textsc{Vertex Disjoint List $\mathcal{C}$-Packing} are W[1]-hard parameterized by pathwidth. (Recall that $\mathcal{P} = \{P_\ell : \ell \ge 1 \}$ and $\mathcal{C} = \{C_\ell : \ell \ge 3\}$.)


\begin{theorem}\label{thm:vertex:treewidth-xp}
    \textsc{Vertex Disjoint List $\mathcal H$-Packing} is solvable in $n^{O(\tw)}$ time, provided that all graphs in $\mathcal H$ are connected, where $n$ is the number of vertices in the input graph.
\end{theorem}
We only sketch the proof of \Cref{thm:vertex:treewidth-xp} by giving a rough idea of a dynamic programming algorithm based on tree decompositions. (A similar idea is used in \cite{Scheffler94} for instance.)
Let $G$ be a graph with $n$ vertices.
We can find a tree decomposition $\mathcal T = (T, \{X_i : i \in V(T)\})$ of $G$ of width $\tw(G)$ in time $n^{O(\tw)}$ using the algorithm of \cite{ArnborgCP87}.
By taking an arbitrary node of $T$, we assume that $T$ is rooted.
For each bag $X_i$ of $\mathcal T$, we denote by $G_i$ the subgraph of $G$ induced by the vertices contained in $X_i$ or descendant bags of $X_i$.

Let $\List$ be a collection of connected subgraphs of $G$, each of which is isomorphic to some graph in $\mathcal H$.
For $i \in V(T)$, a \emph{partial $\mathcal H$-packing} of $G_i$ is a vertex-disjoint subcollection $\mathcal S$ of $\mathcal L_{\mathcal H}$, each of which contains at least one vertex of $G_i$.
A subgraph in a partial $\mathcal H$-packing is said to be \emph{active} if it has at least one vertex in $V(G) \setminus V(G_i)$.
A key observation to our dynamic programming algorithm is that each active subgraph in any partial $\mathcal H$-packing contains at least one vertex of $X_i$.
This follows from the fact that every graph in $\List$ is connected.
This implies that there are at most $|X_i|$ active subgraphs in any partial $\mathcal H$-packing of $G_i$.

To be slightly more precise, we define $\opt(i, \hat{\mathcal S})$ as the maximum cardinality of a partial $\mathcal H$-packing $\mathcal S$ such that $\hat{\mathcal S} \subseteq \mathcal S$ is the set of active subgraphs in it.
By a standard argument in dynamic programming over tree decompositions, we can compute $\opt(i, \hat{\mathcal S})$  in a bottom-up manner.
The total running time of computing all possible $\opt(i, \hat{\mathcal S})$ for $i$ and $\hat{\mathcal S}$ is upper bounded by $|\List|^{O(\tw)}n = n^{O(\tw)}$.

We would like to note that the connectivity of $\mathcal H$ is crucial as we have seen in \Cref{sec:bnddeg} that \textsc{Vertex Disjoint List $H$-packing} is NP-complete even if $G = nK_2$ and $H = 3K_2$.


The following theorems complement the positive result of \Cref{thm:vertex:treewidth-xp}.
\begin{theorem}\label{thm:vertex:treewidth-hardness-solsize}
    \textsc{Vertex Disjoint List $\mathcal{P}$-Packing} is W[1]-hard parameterized by $\pw + k$.
\end{theorem}
\begin{proof}
    The proof is done by showing a parameterized reduction from \textsc{Multicolored Independent Set}, which is known to be W[1]-complete~\cite{FellowsHRV09,CyganFKLMPPS15}.
    In \textsc{Multicolored Independent Set}, we are given a graph $G$ with a partition of vertex set $V(G) = V_1 \cup V_2 \cup \cdots \cup V_t$ and asked whether $G$ has an independent set $S$ of size $t$ that contains exactly one vertex from $V_i$ for each $1 \le i \le t$. 

    From an instance of \textsc{Multicolored Independent Set} $G$ with $V(G) = V_1 \cup \cdots \cup V_t$, we construct a graph $G'$ as follows.
    Let $m = |E(G)|$.
    The vertex set of $G'$ consists of $m + 1$ vertex sets $V'_0 \cup V'_1 \cup \cdots \cup V'_m$, where each set $V'_i$ with $i \ge 1$ is associated with an edge $e_i$ of $G$.
    The vertex set $V'_0$ contains $t$ vertices $v_1, \ldots, v_t$ and other sets $V'_i$ for $i \ge 1$ contain $t + 1$ vertices $v_{i,1}, \ldots, v_{i,t}, w_{i}$.
    For $1 \le i \le m$, we add an edge between every pair of vertices $v \in V_{i-1}$ and $v' \in V_{i}$.
    The graph obtained in this way is denoted by $G'$.
    Observe that the pathwidth of $G'$ is at most $2t + 1$.
    This can be seen by observing that $G'$ has a path decomposition consisting of bags $X_i \coloneqq V'_{i - 1} \cup V'_{i}$ for $1 \le i < m$, which has width at most $2t + 1$.
    For each vertex $v \in V_i$, we define a path $P_v$ in $G'$ that consists of (1) $v_i$ and (2) for all $1 \le j \le m$, $v_{i, j}$ if $e_j$ is not incident to $v$ and $w_i$ otherwise.
    As $P_v$ contains exactly one vertex from each $V'_i$, $P_v$ forms a path in $G'$.
    See~\Cref{fig:red-path-tw} for an illustration.
    \begin{figure}[t]
        \centering
        \includegraphics{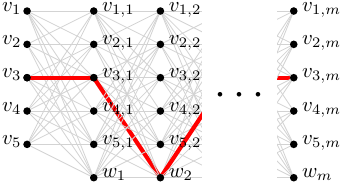}
        \caption{The figure illustrates the graph $G'$. The red thick lines indicate the edges of path $P_v$ for $v \in V_3$. The edge $e_1$ is not incident to $v$ but $e_2$ is incident to $v$.}
        \label{fig:red-path-tw}
    \end{figure}
    Let $\List = \{P_v : v \in V(G)\}$.
    Now, we claim that $G$ has an independent set $S$ that contains exactly one vertex from each $V_i$ if and only if there is a $P_{m+1}$-packing $\mathcal S \subseteq \List$ with $|\mathcal S| \ge t$ in $G'$.

    Suppose that $G$ has an independent set $S$ with $|S \cap V_i| = 1$ for $1 \le i \le t$.
    Let $\mathcal S = \{P_v : v \in S\}$.
    For distinct paths $P_v, P_{v'} \in \mathcal S$ with $v \in V_i$ and $v' \in V_{i'}$, suppose that they share a vertex $u$.
    As $i \neq i'$, the vertex $u$ is contained in some $V_j$ as $u = w_j$.
    This implies that $v$ is adjacent to $v'$, contradicting to the fact that $S$ is an independent set of $G$.

    For the converse direction, suppose that there is a $P_{m+1}$-packing $\mathcal S \subseteq \List$ with $|\mathcal S| \ge t$ in $G'$.
    Let $S = \{v \in V(G) : P_v \in \mathcal S\}$.
    Since each subgraph contains exactly one vertex from $V_0$, we have $|\mathcal S| = t$.
    Moreover, $S$ contains exactly one vertex from each $V_i$.
    If there are two vertices in $S$ that are adjacent (by an edge $e_j$) to each other, the corresponding paths in $\mathcal S$ share the vertex $w_{j}$ in $V_j$.
    Hence, $\mathcal S$ is an independent set of $G$.
\end{proof}

\begin{theorem}\label{thm:vertex:treewidth-hardness-solsize_cycle}
    \textsc{Vertex Disjoint List $\mathcal{C}$-Packing} is W[1]-hard parameterized by $\pw + k$.
\end{theorem}
\begin{proof}
    This proof is done by modifying the instance constructed in the proof of \Cref{thm:vertex:treewidth-hardness-solsize}.
    For the graph $G'$, we additionally give  an edge between every pair of vertices $v \in V'_0$ and $v' \in \cup V'_{m}$.
    Observe that there exists a path decomposition consisting of bags $X_i \coloneqq V'_{i - 1} \cup V'_{i} \cup V'_{m}$ for $1 \le i < m$, which has width at most $3t + 2$.
    For each vertex $v \in V_i$, we define a cycle $C_v$ in $G'$ that consists of (1) $v_i$ and (2) for all $1 \le j \le m$, $v_{i, j}$ if $e_j$ is not incident to $v$ and $w_i$ otherwise.
    The remaining part of the proof is identical to that in \Cref{thm:vertex:treewidth-hardness-solsize}.
\end{proof}

\section{\textsc{Edge Disjoint List \texorpdfstring{$\mathcal H$}{}-Packing} on series-parallel graphs}
\label{sec:series-parallel}
This section is devoted to showing several positive and negative results on \emph{series-parallel graphs}.
The class of series-parallel graphs is a well-studied class of graphs and is equivalent to the class of graphs of treewidth at most $2$.

A \emph{two-terminal labeled graph} is a graph $G$ with distinguished two vertices called a \emph{source} $s$ and a \emph{sink} $t$.
Let $G_1 = (V_1, E_1)$ (resp.\ $G_2 = (V_2, E_2)$) be a two-terminal labeled graph with a source $s_1$ and a sink $t_1$ (resp.\ a source $s_2$ and a sink $t_2$).
A \emph{series composition} of $G_1$ and $G_2$ is an operation that produces the two-terminal labeled graph with a source $s$ and a sink $t$ obtained from $G_1$ and $G_2$ by identifying $t_1$ and $s_2$, where $s = s_1$ and $t = t_2$.
A \emph{parallel composition} of $G_1$ and $G_2$ is an operation that produces the two-terminal labeled graph with a source $s$ and a sink $t$ obtained from $G_1$ and $G_2$ by identifying $s_1$ and $s_2$, and identifying $t_1$ and $t_2$, where $s = s_1 (= s_2)$ and $t = t_1 (= t_2)$.
We denote $G = G_1 \sericom G_2$ if $G$ is created by a series composition of $G_1$ and $G_2$, and denote $G = G_1 \paracom G_2$ if $G$ is created by a parallel composition of $G_1$ and $G_2$.
We say that a two-terminal labeled graph $G$ is a \emph{two-terminal series-parallel graph} if one of the following conditions is satisfied: (i) $G = K_2$ with a source $s$ and a sink $t$; (ii) $G = G_1 \sericom G_2$ for two-terminal series-parallel graphs $G_1$ and $G_2$; or (iii) $G = G_1 \paracom G_2$ for two-terminal series-parallel graphs $G_1$ and $G_2$. 

We say that a graph $G$ (without a source and a sink) is a \emph{series-parallel graph} if each biconnected component is a two-terminal series-parallel graph by regarding some two vertices as a source and a sink\footnote{Some papers refer to a two-terminal series-parallel graph simply as a series-parallel graph. In this paper, we distinguish them explicitly to avoid confusion.}.

\subsection{Hardness}



A graph $G$ is \emph{outerplanar} if it has a planar embedding such that every vertex of $G$ meets the unbounded face of the embedding.
Every outerplanar graph is series-parallel but may not be two-terminal series-parallel.
The following two theorems indicate that \textsc{Edge Disjoint List $\mathcal{H}$-Packing} remains intractable even when a given graph is highly restricted.

\begin{theorem}\label{thm:w1-hardness_path}
    \textsc{Edge Disjoint List $\mathcal{P}$-Packing} parameterized by $\bandw{G}+k$ is W[1]-hard even for outerplanar and two-terminal series-parallel graphs, where $k$ is a solution size.
\end{theorem}

We prove Theorem~\ref{thm:w1-hardness_path} by giving a parameterized reduction from \textsc{Multicolored Independent Set} to \textsc{Edge Disjoint List $\mathcal{P}$-Packing}.
Recall that, in \textsc{Multicolored Independent Set}, we are given a graph $G^\prime$ with a partition of its vertex set $V(G^\prime) = V_1 \cup V_2 \cup \cdots \cup V_t$.
We may assume that $V_i$ forms a clique of $G^\prime$.

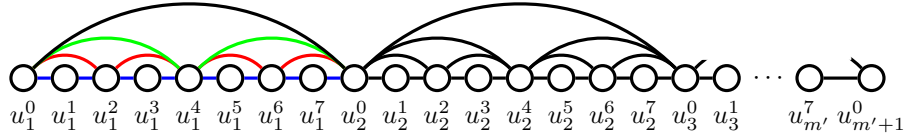
\begin{figure}[h]
	\centering
	\begin{tikzpicture}
		\newcommand{\edgelength}{0.55}
		\newcounter{pathnode}
		\newcounter{pathedge}
		\newcounter{tempvalue}
		\setcounter{pathnode}{8}
		\setcounter{pathedge}{\value{pathnode}}
        \addtocounter{pathnode}{-1}
		\addtocounter{pathedge}{-2}
		
		\foreach \x in {0,...,\value{pathnode}} \node[draw, circle, very thick] at (\x*\edgelength, 0)[label = below:$u_1^{\x}$] (u_1^\x){};
		\foreach \x in {0,...,\value{pathnode}} \node[draw, circle, very thick] at (\x*\edgelength+8*\edgelength, 0)[label = below:$u_2^{\x}$] (u_2^\x){};
		\foreach \x in {0,1} \node[draw, circle, very thick] at (\x*\edgelength+16*\edgelength, 0)[label = below:$u_3^{\x}$] (u_3^\x){};
		\node[circle, very thick] at (18*\edgelength, 0) (u_3^2)[]{$\mathbb{\cdots}$};
		\node[draw, circle, very thick] at (19*\edgelength, 0)[label = below:$u_{m^\prime}^{7}$] (u_m^7){};
		\node[draw, circle, very thick] at (20.5*\edgelength, 0)[label = below:$u_{m^\prime+1}^0$] (u_m+1^0){};
		
		\foreach \x in {0,...,\value{pathedge}} {
			\setcounter{tempvalue}{\x};
			\addtocounter{tempvalue}{1};
			\draw [very thick,blue] (u_1^\x) -- (u_1^\arabic{tempvalue});
			\draw [very thick] (u_2^\x) -- (u_2^\arabic{tempvalue});
		}
		\draw [very thick,blue] (u_1^7) -- (u_2^0);
		\draw [very thick] (u_2^7) -- (u_3^0) -- (u_3^1);
        \draw [very thick] (u_m^7) -- (u_m+1^0);
		
		\foreach \x in {0, 2, 4} {
			\setcounter{tempvalue}{\x};
			\addtocounter{tempvalue}{2};
			\draw [very thick,red] (u_1^\x) to [out = 45, in = 135] (u_1^\arabic{tempvalue});
			\draw [very thick] (u_2^\x) to [out = 45, in = 135] (u_2^\arabic{tempvalue});
		}
		\draw [very thick,red] (u_1^6) to [out = 45, in = 135] (u_2^0);
		\draw [very thick] (u_2^6) to [out = 45, in = 135] (u_3^0);
		
		\draw [very thick,green] (u_1^0) to [out = 45, in = 135] (u_1^4);
		\draw [very thick,green] (u_1^4) to [out = 45, in = 135] (u_2^0);
		\draw [very thick] (u_1^0) to [out = 45, in = 135] (u_2^0);
		
		\draw [very thick] (u_2^0) to [out = 45, in = 135] (u_2^4);
		\draw [very thick] (u_2^4) to [out = 45, in = 135] (u_3^0);
		\draw [very thick] (u_2^0) to [out = 45, in = 135] (u_3^0);
		
		\draw [very thick] (u_3^0) to [out = 45, in = 135] (u_m+1^0);
		
		\filldraw [white, fill=white] ($(u_3^2) + (-1.1,1)$) rectangle ($(u_3^2) + (1.1,0.25)$);
	
	\end{tikzpicture}
	\caption{The figure partially illustrates the graph $G$ for $t=3$.  The edges in $E_1^1$, in $E_1^2$, and in $E_1^3$ are depicted as green segments, red segments, and blue segments, respectively.}
	\label{fig:hardness_PSP_outerplanar}
\end{figure}

From an instance of \textsc{Multicolored Independent Set}, we construct an instance $(G, t, \List)$ of \textsc{Edge Disjoint List $\mathcal{P}$-Packing}.
To this end, we construct a gadget $G_i$ for each $i\in \intsec{m^\prime}$, where $m^\prime = |E(G^\prime)|$. 
(See also \figureautorefname~\ref{fig:hardness_PSP_outerplanar}.)
The vertex set of $G_i = (W_i, E_i)$ is defined as $W_i = \{u_i^j : j \in \intsecz{2^{t}} \}$.
For $j \in \intsecz{t}$ and $p \in \intsecz{2^{j}}$, let denote $\landnum{j}{p} = 2^{t-j}p$.
Then we define $E_i^j = \{u_i^{\landnum{j}{p-1}} u_i^{\landnum{j}{p}} : p \in \intsec{2^{j}}\}$.
In other words, $E_i^j$ consists of all edges in the path starting from $u_i^0$ to $u_i^q$ with $q = 2^{t}$ that contains $u_i^{r}$ for $r = 2^{t-j}, 2\cdot 2^{t-j}, 3\cdot 2^{t-j}, \ldots, (2^{j}-1)2^{t-j}$ as internal vertices.
We define the edge set of $G_i$ as $E_i = \bigcup_{j \in \intsecz{t}} E_i^j$.
After constructing $G_1, \ldots, G_{m^\prime}$, we concatenate the graphs by identifying $u_i^q$ of $G_i$ and $u_{i+1}^0$ of $G_{i+1}$ for every $i\in \intsec{m^\prime-1}$.
For notational convenience, we relabel a vertex $u_i^q$ as $u_{i+1}^0$ for every $i \in \intsec{m^\prime}$.
We define the graph constructed above as $G$.

Next, we construct a collection $\List$ of paths in $G$.
Suppose that $E(G^\prime) = \{ e_1, e_2, \ldots, e_{m^\prime} \}$.
For each $v \in V(G^\prime)$, a path $P_v$ of $G$ is defined as follows.
Suppose that $v \in V_j$.
The path $P_v$ starts at $u_1^0$ and ends at $u_{m^\prime+1}^0$.
For each $i \in \intsec{m^\prime}$, $P_v$ passes through $u_i^0 u_{i+1}^0$ if $e_i$ is incident to $v$; otherwise, $P_v$ passes through all edges in $E_i^j$.
For example, consider $t = 3$ and a vertex $v \in V_2$ such that only edges $e_1, e_4$ are incident to $v$. Then, 
\begin{align*}
	P_v = (u_1^0, u_2^0, u_2^2, u_2^4, u_2^6, u_3^0, u_3^2, u_3^4, u_3^6, u_4^0, u_5^0, u_5^2, \ldots, u_{m^\prime}^6, u_{m^\prime+1}^0).
\end{align*}
We define $\List = \{ P_v : v \in V(G^\prime) \}$.
This completes the construction of the instance $(G, t, \List)$.
Let denote $n = |V(G)|$ and $m = |E(G)|$.

Observe that $G$ is an outerplanar graph and a two-terminal series-parallel graph.
For $n^\prime = |V(G^\prime)|$ and $m^\prime = |E(G^\prime)|$, we have $n = 2^{t} m^\prime + 1$, $m = (2^{t+1} - 1)m^\prime$, and $|\List| = n^\prime$.
Thus, the construction of $(G, k, \List)$ is completed in $O(2^{t} n^\prime m^\prime)$ time.
Moreover, $\bandw{G} \le 2^{t}$ holds: consider the linear layout $\pi(u_i^j) = 2^{t}(i-1) + j+1$ for $i\in \intsec{m^\prime}$ and $j \in \intsec{2^t-1}\cup \{0\}$, which is the same as the ordering of the vertices depicted in \figureautorefname~\ref{fig:hardness_PSP_outerplanar}.
The following lemma completes the proof of \Cref{thm:w1-hardness_path}.

\begin{lemma} \label{lem:reduction_w1-hard_path}
    The instance $(G^\prime, t, \{ V_1, V_2, \ldots, V_{t} \})$ of \textsc{Multicolored Independent Set} is a yes-instance if and only if the instance $(G, t, \List)$ of \textsc{Edge Disjoint List $\mathcal{P}$-Packing} is a yes-instance.
\end{lemma}

\begin{proof}
    We first prove the necessity. 
    For a solution $\mathcal{S}$ of a yes-instance $(G, t, \List)$ of \textsc{Edge Disjoint List $\mathcal{P}$-Packing}, let $I = \{ v \in V(G^\prime) : P_v \in \mathcal{S}\}$.
    Since $\mathcal{S}$ are mutually edge-disjoint on $G$, at most one path in $\mathcal{S}$ passes through $u_{i}^0 u_{i+1}^0$ for each $i \in \intsec{m^\prime}$.
    This implies that at most one of the endpoints of $e_i \in E(G^\prime)$ is in $I$, that is, $I$ is an independent set of $G^\prime$.
    Moreover, we have $|I| \ge t$.
    Since $V_j$ is a clique of $G^\prime$, $|V_j \cap I| = 1$ holds for every $j \in \intsec{t}$. 
    Thus, $(G^\prime, t, \{ V_1, V_2, \ldots, V_{t} \})$ is a yes-instance of \textsc{Multicolored Independent Set}.

    We next prove the sufficiency.
    For a solution $I \subseteq V(G^\prime)$ of a yes-instance $(G^\prime, t, \{ V_1, V_2, \ldots, V_{t} \})$ of \textsc{Multicolored Independent Set}, let $\mathcal{S} = \{ P_v : v \in I \}$.
    Since $|\mathcal{S}| = t $, it suffices to show that $\mathcal{S}$ is mutually edge-disjoint.
    Let $u, v \in I$. Suppose that $u \in V_{j}$ and $v \in V_{j^\prime}$.
    Since $|V_j \cap I| = 1$ for every $j \in \intsec{t}$, $j \neq j'$ holds.
    For paths $P_{u}, P_{v} \in \mathcal{S}$ corresponding to $u, v \in I$ and $i \in \intsec{m^\prime}$, we denote by $P_{u,i}$ and $P_{v,i}$, the subpaths of $P_{u}$ and $P_{v}$ that start at $u_i^0$ and end at $u_{i+1}^0$, respectively.
    In the remainder of this proof, we show that $P_{u,i}$ and $P_{v,i}$ are edge-disjoint for every $i \in \intsec{m^\prime}$, meaning that $P_{u}$ and $P_{v}$ are also edge-disjoint.
    Suppose that $e_i$ is incident to neither $u$ nor $v$.
    Then, it holds that $E(P_{u,i}) = E_i^{j}$ and $E(P_{v,i}) = E_i^{j'}$. 
    Since $j \neq j'$, we have $E_i^{j} \cap E_i^{j} = \emptyset$.
    Suppose otherwise that $e_i$ is incident to $u$ or $v$, say $u$.
    This implies that $e_i$ is not incident to $v$ as $u,v \in I$.
    Therefore, we have $E(P_{u}) = E_i^{0} $ and $E(P_{v,i}) = E_i^{j'}$.
    Since $j' \in \intsec{t}$, $P_{u,i}$ and $P_{v,i}$ are edge-disjoint. 
\end{proof}

By a similar proof, we can show the W[1]-hardness of \textsc{Edge Disjoint List $\mathcal{C}$-Packing}.

\begin{theorem}\label{thm:w1-hardness_cycle}
    \textsc{Edge Disjoint List $\mathcal{C}$-Packing} parameterized by $\bandw{G}+k$ is W[1]-hard even for outerplanar and two-terminal series-parallel graphs, where $k$ is a solution size.
\end{theorem}

\begin{proof}
    We reuse the instance $(G, k, \List)$ of \textsc{Edge Disjoint List $\mathcal{P}$-Packing} from the proof of \Cref{thm:w1-hardness_path}, which is constructed from the instance $(G^\prime, t, \{ V_1, V_2, \ldots, V_{t} \})$ of \textsc{Multicolored Independent Set}.
    Let $G^c$ denote the graph obtained by identifying $u_1^0$ and $u_{m+1}^0$ of $G$.
    Observe that $G^c$ is outerplanar and two-terminal series-parallel.
    Moreover, this identification produces a collection $\List^c$ of cycles of $G^c$ from $\List$, because every path in $\List$ starts at $u_1^0$ and ends at $u_{m+1}^0$.
    As in the proof of Lemma~\ref{lem:reduction_w1-hard_path}, we can show that the instance $(G^\prime, t, \{ V_1, V_2, \ldots, V_{t} \})$ of \textsc{Multicolored Independent Set} is a yes-instance if and only if the instance $(G^c, k, \List^c)$ of \textsc{Edge Disjoint List $\mathcal{C}$-Packing} is a yes-instance.

    In the remainder of this proof, we show that $\bandw{G^c} \le 3\cdot 2^{t}$.
    For each $i \in \intsec{m^\prime}$, let $\pi_i = \langle u_i^0, u_i^1,\ldots, u_i^{q-1} \rangle$, where $q = 2^{t}$.
    In addition, we define a permutation $\rho$ of $\intsec{m^\prime}$ as follows:
	\begin{align*}
		\rho(i) = 
		\begin{cases}
			\frac{i+1}{2} & \text{if $i$ is odd}, \\
			m^\prime + 1 - \frac{i}{2} & \text{otherwise.}
		\end{cases}
	\end{align*}
    Consider an ordering $\pi = \langle \pi_{\rho(1)}, \pi_{\rho(2)}, \ldots, \pi_{\rho(m^\prime)}  \rangle$.
    Recall that $W_i = \{u_i^j : j \in \intsec{2^{t}} \cup \{ 0\} \}$ and $u_i^{2^t} = u_{i+1}^0$ for the gadget $G_i$.
    From the construction of $G^c$, every vertex in $W_i$ for $i \in \intsec{m^\prime}$ is adjacent to only vertices in $W_{i-1}\cup W_i \cup W_{i+1}$, where we consider $W_{m^\prime} = W_0$ and $W_1 = W_{m^\prime + 1}$.
    Thus, for every vertex $u$ in $\pi_{\rho(i)}$, each adjacent vertex of $u$ appears between $\pi_{\rho(i-2)}$ and $\pi_{\rho(i+2)}$.
    Since $\pi_{\rho(i)}$ has $2^{t}$ vertices, the bandwidth of $\pi$ is bounded by $3\cdot 2^{t}$, that is, $\bandw{G^c} \le 3\cdot 2^{t}$.
    This completes the proof of Theorem~\ref{thm:w1-hardness_cycle}. 
\end{proof}


We next focus on the case where $\mathcal{H}$ consists of a single graph and show that the problem remains hard.
Let $K_{2,n}$ denotes the complete bipartite graph such that one side consists of two vertices and the other side consists of $n$ vertices.

\begin{theorem}\label{thm:sp_P4}
    \textsc{Edge Disjoint List $P_4$-packing} remains NP-complete even for the class of $K_{2,n}$.
\end{theorem}

Observe that $K_{2,n}$ is a two-terminal series-parallel graph.
Since \textsc{Edge Disjoint List $P_3$-packing} is solvable for general graphs, Theorem~\ref{thm:sp_P4} suggests that the complexity dichotomy with respect to path length still holds for very restricted graphs.
Moreover, Theorem~\ref{thm:sp_P4} immediately provides the following corollary, which strengthens the hardness result in~\cite{AravindS23} that \textsc{Path Set Packing} is W[1]-hard when parameterized by vertex cover number of $G$ plus maximum length of paths in a given collection $\mathcal{L}$.

\begin{corollary}
    \textsc{Path Set Packing} is NP-complete even when a given graph has vertex cover number $2$ and every path in $\mathcal{L}$ is of length $3$.
\end{corollary}

To prove Theorem~\ref{thm:sp_P4}, we perform a polynomial-time reduction from \textsc{Independent Set} on cubic graphs, that is, graphs such that every vertex has degree exactly $3$.
\textsc{Independent Set} on cubic graphs is known to be NP-complete~\cite{GareyJS76}.

For a cubic graph $G^\prime$, we construct vertex sets $X = \{x_u, x_u^\prime : u \in V(G^\prime) \}$, $Y = \{y_{u,e}, y_{v,e}: e=uv\in E(G^\prime) \}$, and $Z = \{z_{e} : e \in E(G^\prime) \}$. 
A graph $G$ consists of a vertex set $X\cup Y \cup Z \cup \{s,t\}$ such that $X\cup Y \cup Z$ forms an independent of $G$ and every vertex in $X\cup Y \cup Z$ is adjacent to both $s$ and $t$.
Obviously, $G = K_{2, 2n^\prime + 3m^\prime}$.

We next construct a collection $\List$ of paths in $G$ as follows:
\begin{itemize}
    \item for $u \in V(G^\prime)$, let $P_u^X = \langle s, x_u, t, x_u^\prime \rangle$;
    \item for $u \in V(G^\prime)$ with three incident edges $e, f, g \in \incedge_{G'}(u)$, let 
    \begin{align*}
        P_{u, e}^Y &= \langle s, y_{u,e}, t, x_u^\prime \rangle,\\
        P_{u, f}^Y &= \langle s, y_{u,f}, t, x_u \rangle,\\
        P_{u, g}^Y &= \langle x_u, s, y_{u,g}, t \rangle;
    \end{align*}
    \item for $e = uv \in E(G')$, let $P_{u, e}^Z = \langle y_{u, e}, s, z_{e}, t \rangle$ and $P_{v, e}^Z = \langle y_{v, e}, s, z_{e}, t \rangle$.
\end{itemize}
Then, we define
\begin{align*}
    \List = \{ P_u^X  : u \in V(G^\prime) \} \cup \{ P_{u, e}^Y : u \in V(G'), e \in \incedge_{G'}(u) \}\cup \{ P_{u, e}^Z, P_{v, e}^Z: e = uv\in E(G') \}
\end{align*} and $k = k^\prime + 2m^\prime$.
Clearly, the construction of the instance $(G, k, \List)$ for \textsc{Edge Disjoint List $P_4$-packing} is completed in polynomial time.
Our remaining task is to prove the following lemma.

\begin{lemma} \label{lem:sp_P4_correctness}
    An instance $(G^{\prime}, k^\prime)$ of \textsc{Independent Set} is a yes-instance if and only if an instance $(G,k,\List)$ of \textsc{Edge Disjoint List $P_4$-packing} is a yes-instance.
\end{lemma}

\begin{proof}
    We first show the forward implication.
    For an independent set $I$ with $|I| \ge k^\prime$, we construct a subcollection $\mathcal{S} \subseteq \List$ as follows:
    \begin{align*}
        \mathcal{S} = \{P_u^X : u \in I\} \cup \{P_{u,e}^Z: u \in I, e \in \incedge_{G'}(u) \}
                      \cup \{P_{u,e}^Y : u \notin I, e \in \incedge_{G'}(u) \}.
    \end{align*}

    Observe that $|\mathcal{S}| \ge k^\prime + 3n^\prime$.
    Since $G^\prime$ is a cubic graph, we have $3n^\prime = 2m^\prime$ by the handshaking lemma, and hence $|\mathcal{S}| \ge k^\prime + 2m^\prime = k$.
    We claim that $\mathcal{S}$ is mutually edge-disjoint.
    Assume for a contradiction that there are two paths $P$ and $P^\prime$ in $\mathcal{S}$ that have a common edge.
    There are three cases to consider: (i) $P = P_u^X$ for some $u \in V(G^\prime)$; (ii) $P = P_{u, e}^Y$ for some $u \in V(G^\prime)$ and $e \in \incedge_{G'}(u)$; and (iii) $P = P_{u, e}^Z$ and $P^\prime = P_{v,f}^Z$ for some $u,v \in V(G^\prime)$, $e \in \incedge_{G'}(u)$, and $f \in \incedge_{G'}(v)$.
    
    In the case (i), from the construction of $\List$, we have $P^\prime = P_{u,e}^Y$ for some $e \in \incedge_{G'}(u)$.
    As $P_u^X \in \mathcal S$, we have $u \in I$ but as $P_{u, e}^Y \in \mathcal S$, we have $u \notin I$, a contradiction.
    Consider the case (ii). 
    Only $P_u^X$ and $P_{u,e}^Z$ can share an edge with $P_{u,e}^Y$, and thus we consider the latter case: $P^\prime = P_{u,e}^Z$. 
    This implies that $u \in I$, while $P_{u,e}^Y \in \mathcal{S}$ also implies $u \notin I$, a contradiction.
    In the case (iii), one can observe that $e = f = uv$.
    Then, we have $u, v \in I$, a contradiction.
    Therefore, $\mathcal{S}$ is a solution for $(G,k,\List)$.

    We next show the backward implication.
    Let $\mathcal{S}$ be a subcollection of $\List$ such that $\mathcal{S}$ is mutually edge-disjoint and $|\mathcal{S}| \ge k = k^\prime + 2m^\prime$.
    Suppose that $\mathcal{S}$ contains both $P_u^X$ and $P_v^X$ with $e = uv \in E(G^\prime)$.
    As they share edges with $P_{u,e}^Y$ and $P_{v,e}^Y$, we have $P_{u,e}^Y, P_{v,e}^Y \notin \mathcal{S}$.
    Furthermore, at least one of $P_{u,e}^Z$ and $P_{v,e}^Z$ is not contained in $\mathcal{S}$.
    We then remove $P_u^X$ from $\mathcal{S}$ and add $P_{u,e}^Y$ into $\mathcal{S}$ if $P_{u,e}^Z \notin \mathcal{S}$; otherwise, we remove $P_v^X$ from $\mathcal{S}$ and add $P_{v,e}^Y$ into $\mathcal{S}$.
    This exchange keeps $\mathcal{S}$ mutually edge-disjoint because $P_{u,e}^Y$ (resp. $P_{v,e}^Y$) shares edges only with $P_u^X$ (resp. $P_v^X$) in $\mathcal S$.
    Hence, we may assume that $\mathcal{S}$ contains at most one of $P_u^X$ and $P_v^X$ for every $uv \in E(G^\prime)$.
    Let denote $\mathcal{P}_{e} = \{ P_{u,e}^Y, P_{u,e}^Z, P_{v,e}^Y, P_{v,e}^Z\}$ for $e = uv \in E(G')$.
    Observe that $|\mathcal{P}_{e} \cap \mathcal{S}| \le 2$ due to the construction of $\List$.
    Thus, $|\bigcup_{e \in E(G^\prime)}\mathcal{P}_{e} \cap \mathcal{S} | \le 2m^\prime$.
    This implies that $| \{ P_u^X : u \in V(G^\prime)\}\cap \mathcal{S}| \ge k - 2m^\prime = k^\prime$.
    Let $I = \{ u \in V(G^\prime) : P_u^X \in \mathcal{S} \}$.
    From the assumption of $\mathcal{S}$, $I$ is an independent set of $G'$.
    This completes the proof of the lemma. 
\end{proof}

We also show the complexity of \textsc{Edge Disjoint List $C_5$-packing}, which highlights the positive result in \Cref{subsec:sp_4cycle}.

\begin{theorem}\label{thm:sp_C5}
    \textsc{Edge Disjoint List $C_5$-packing} remains NP-complete even for two-terminal series-parallel graphs.
\end{theorem}

We again perform a polynomial-time reduction from \textsc{Independent Set} on cubic graphs.
For a cubic graph $G^\prime$, we use a gadget $G_u^X$ depicted in \figureautorefname~\ref{fig:hardness_sp_C5} for each $u \in V(G^\prime)$.
\begin{figure}[t]
	\centering
	\begin{tikzpicture}
		\newcommand{\elen}{1.75}
        \newcommand{\eheight}{1.25}
        \node[draw, circle, very thick] at (0, 0)[label = below:$x_u$] (x){};
        \node[draw, circle, very thick] at (-\elen, 0)[label = below:$s_u^X$] (s){};
        \node[draw, circle, very thick] at (\elen, 0)[label = below:$t_u^X$] (t){};
        \node[draw, circle, very thick] at ($(s)!0.5!(x) + (0,\eheight)$)[label = above:$x_{u,e}$] (x_1){};
        \node[draw, circle, very thick] at ($(x)!0.5!(t) + (0,\eheight)$)[label = above:$x_{u,f}$] (x_2){};
        \node[draw, circle, very thick] at ($(s)!0.25!(t) + (0,-\eheight)$)[label = below:$x_{u,g}$] (x_31){};
        \node[draw, circle, very thick] at ($(s)!0.75!(t) + (0,-\eheight)$)[label = below:$x_{u,g}^\prime$] (x_32){};

        \draw [very thick] (s) -- (x) -- (t);
        \draw [very thick] (s) -- (x_1) -- (x) -- (x_2) --(t) -- (x_32) -- (x_31) -- (s);
	\end{tikzpicture}
	\caption{The gadget $G_u^X$ for $u \in V(G^\prime)$ with three incident edges $e,f,g \in \incedge_{G^\prime}(u)$.}
	\label{fig:hardness_sp_C5}
\end{figure}
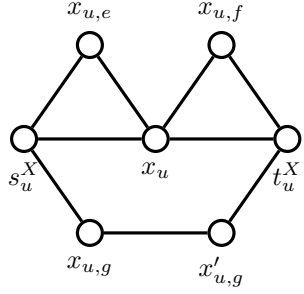
The gadget $G_u^X$ consists of seven vertices $s_u^X$, $t_u^X$, $x_u$, $x_{u,e}$, $x_{u,f}$, $x_{u,g}$, and $x_{u,g}^\prime$, where $e,f,g \in \incedge_{G'}(u)$.
Observe that $G_u^X$ is a two-terminal series-parallel graph with a source $s_u^X$ and a sink $t_u^X$. 
For each pair $u \in E(G^\prime)$ and $e \in \incedge_{G^\prime}(u)$, we prepare an additional two-terminal series-parallel graph $G_{u,e}^Y$ consisting of a path $\langle s_{u,e}^Y, y_{u,e}, t_{u,e}^Y \rangle$, where $s_{u,e}^Y$ and $t_{u,e}^Y$ are a source and a sink of $G_{u,e}^Y$, respectively.
Moreover, for each $e \in E(G^\prime)$, we construct a two-terminal series-parallel graph $G_{e}^Z$ consisting of a path $\langle s_{e}^Z, z_{e}, z_{e}^\prime, t_{e}^Z \rangle$, where $s_{e}^Z$ and $t_{e}^Z$ are a source and a sink of $G_{e}^Z$, respectively.
We denote by $G^X$ a graph obtained by parallel compositions of the all gadgets $G_u^X$ for $u \in V(G^\prime)$.
The graph $G^Y$ and $G^Z$ are similarly defined as parallel compositions of the gadgets.
We then let $G = G^X \paracom G^Y \paracom G^Z$.
Observe that $G$ is also a two-terminal series-parallel graph.
Let $s$ and $t$ denote the source and the sink of $G$, respectively.

We next construct a collection $\List$ of cycles in $G$ as follows:
\begin{itemize}
    \item for $u \in V(G^\prime)$, let $C_u^X = \langle s, x_u, t, x_{u,g}^\prime, x_{u,g}, s \rangle$, where $g \in \incedge_{G^\prime}(u)$;
    \item for $u \in V(G^\prime)$ with three incident edges $e, f, g \in \incedge_{G^\prime}(u)$, let 
    \begin{align*}
        C_{u,e}^Y &= \langle s, x_{u,e}, x_u, t, y_{u,e}, s \rangle,\\
        C_{u,f}^Y &= \langle s, x_u, x_{u,f}, t, y_{u,f}, s \rangle,\\
        C_{u,g}^Y &= \langle s, x_{u,g}, x_{u,g}^\prime, t, y_{u,g}, s \rangle;
    \end{align*}
    \item for $ e = uv \in E(G^\prime)$, let $C_{u,e}^Z = \langle s, y_{u,e}, t, z_{e}^\prime, z_{e}, s \rangle$ and $C_{v,e}^Z = \langle s, y_{v,e}, t, z_{e}^\prime, z_{e}, s \rangle$.
\end{itemize}
Then, we define
\begin{align*}
    \List = \{ C_u^X  : u \in V(G^\prime) \} \cup \{ C_{u, e}^Y : u \in V(G'), e \in \incedge_{G'}(u) \}\cup \{ C_{u, e}^Z, C_{v, e}^Z: e = uv\in E(G') \}
\end{align*} and $k = k^\prime + 2m^\prime$, where $m' = |E(G')|$.
The construction of the instance $(G, k, \List)$ for \textsc{Edge Disjoint List $C_5$-packing} takes in polynomial time.
From a vertex set $I$ of $G'$, we can construct a set
\begin{align*}
    \mathcal{S} = \{C_u^X : u \in I\} \cup \{C_{u,e}^Z: u \in I, e \in \incedge_{G'}(u) \}
                    \cup \{C_{u,e}^Y : u \notin I, e \in \incedge_{G'}(u) \}
\end{align*}
of cycles in $\List$.
Analogous to Lemma~\ref{lem:sp_P4_correctness}, we can prove that $I$ is an independent set of $G'$ if and only if cycles in $\mathcal S$ are edge-disjoint, which is omitted because, in fact, it is essentially the same as the proof of Lemma~\ref{lem:sp_P4_correctness}.


\subsection{Polynomial-time algorithm of \textsc{Edge Disjoint List \texorpdfstring{$C_\ell$}{}-packing} for \texorpdfstring{$\ell \le 4$}{}} \label{subsec:sp_4cycle}

We design a polynomial-time algorithm for \textsc{Edge Disjoint List $C_\ell$-packing} for $\ell \le 4$ on two-terminal series-parallel graphs.
Actually, we give a stronger theorem.

\begin{theorem} \label{thm:sp_4cycle}
    Let $\mathcal{C}_{\le 4} = \{C_3, C_4\}$.
    Given a series-parallel graph $G$ with $n$ vertices and a collection $\List$ of cycles in $G$ of length at most $4$, \textsc{Edge Disjoint List $\mathcal{C}_{\le 4}$-packing} is solvable in $O(|\List| + n^{2.5})$ time.
\end{theorem}

We first note that
we may assume that a given graph $G$ is biconnected: the problem can be solved independently in each biconnected component.
Moreover, from the definition of series-parallel graphs, every biconnected series-parallel graph can be regarded as a two-terminal series-parallel graph.
We thus consider a polynomial-time algorithm that finds a largest solution of a given two-terminal series-parallel graph.

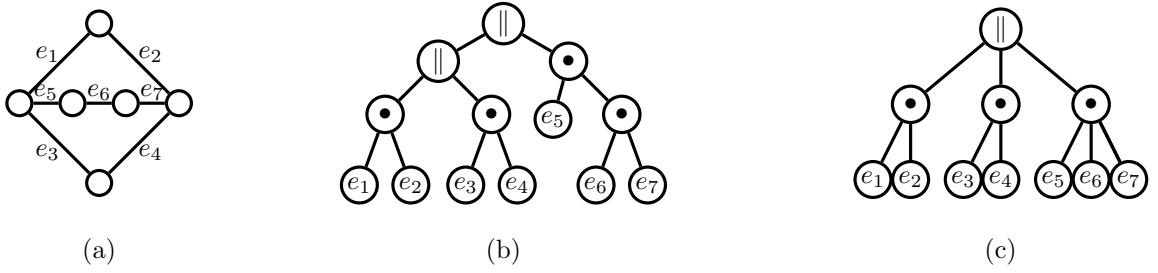
\begin{figure}[t]
    \centering
    \begin{tabular}{ccc}
        
      \begin{minipage}{0.25\textwidth}
          \centering
          \begin{tikzpicture}
                \def\edgew{1.5}
                \node[draw, circle, very thick] at (0,0) (v1){};
                \node[draw, circle, very thick] at ($(v1) + (-135:\edgew)$) (v2){};
                \node[draw, circle, very thick] at ($(v1) + (-45:\edgew)$) (v3){};
                \node[draw, circle, very thick] at ($(v2) + (-45:\edgew)$) (v4){};
                \node[draw, circle, very thick] at ($(v2)!0.333!(v3)$) (v5){};
                \node[draw, circle, very thick] at ($(v3)!0.333!(v2)$) (v6){};
                \node[] at ($(v2) + (60:0.5*\edgew)$) (e1){$e_1$};
                \node[] at ($(v3) + (120:0.5*\edgew)$) (e2){$e_2$};
                \node[] at ($(v2) + (-60:0.5*\edgew)$) (e3){$e_3$};
                \node[] at ($(v3) + (-120:0.5*\edgew)$) (e4){$e_4$};
                \node[] at ($(v2)!0.5!(v5) + (0,0.15)$) (e5){$e_5$};
                \node[] at ($(v5)!0.5!(v6) + (0,0.15)$) (e6){$e_6$};
                \node[] at ($(v6)!0.5!(v3) + (0,0.15)$) (e7){$e_7$};
                
                \draw [very thick] (v2) -- (v1) -- (v3) -- (v4) -- (v2) -- (v5) -- (v6) -- (v3);
            \end{tikzpicture}
      \end{minipage} &
      
			
			

      \begin{minipage}{0.35\textwidth}
          \centering
          \begin{tikzpicture}
            \def\edgew{1}
            \def\degl{-110}
            \def\degr{-70}
            \def\insepe{1pt}
            \def\insepseri{2.5pt}
            \def\inseppara{1.5pt}
            \node[draw, circle, very thick,inner sep=\inseppara] at (0,0) (t1){$\paracom$};
            
            \node[draw, circle, very thick,inner sep=\inseppara] at ($(t1) + (-150:\edgew)$) (t2){$\paracom$};
            \node[draw, circle, very thick,inner sep=\insepseri] at ($(t1) + (-30:\edgew)$) (t3){$\sericom$};
            
            \node[draw, circle, very thick,inner sep=\insepseri] at ($(t2) + (-135:\edgew)$) (t4){$\sericom$};
            \node[draw, circle, very thick,inner sep=\insepseri] at ($(t2) + (-45:\edgew)$) (t5){$\sericom$};
            \node[draw, circle, very thick,inner sep=\insepe] at ($(t3) + (-105:0.8*\edgew)$) (t6){$e_5$};
            \node[draw, circle, very thick,inner sep=\insepseri] at ($(t3) + (-45:\edgew)$) (t7){$\sericom$};
            
            \node[draw, circle, very thick,inner sep=\insepe] at ($(t4) + (\degl:\edgew)$) (t8){$e_1$};
            \node[draw, circle, very thick,inner sep=\insepe] at ($(t4) + (\degr:\edgew)$) (t9){$e_2$};
            \node[draw, circle, very thick,inner sep=\insepe] at ($(t5) + (\degl:\edgew)$) (t10){$e_3$};
            \node[draw, circle, very thick,inner sep=\insepe] at ($(t5) + (\degr:\edgew)$) (t11){$e_4$};
            \node[draw, circle, very thick,inner sep=\insepe] at ($(t7) + (\degl:\edgew)$) (t12){$e_6$};
            \node[draw, circle, very thick,inner sep=\insepe] at ($(t7) + (\degr:\edgew)$) (t13){$e_7$};
    
            \draw [very thick] (t8) -- (t4) -- (t2) -- (t1) -- (t3) -- (t7) -- (t13);
            \draw [very thick] (t4) -- (t9);
            \draw [very thick] (t10) -- (t5) -- (t11);
            \draw [very thick] (t2) -- (t5);
            \draw [very thick] (t3) -- (t6);
            \draw [very thick] (t7) -- (t12);
        \end{tikzpicture} 
      \end{minipage} &

      \begin{minipage}{0.4\textwidth}
          \centering
          \begin{tikzpicture}
            \def\edgew{1}
            \def\degl{-110}
            \def\degr{-70}
            \def\insepe{1pt}
            \def\insepseri{2.5pt}
            \def\inseppara{1.5pt}
            \node[draw, circle, very thick,inner sep=\inseppara] at (0,0) (l1){$\paracom$};
            
            \node[draw, circle, very thick,inner sep=\insepseri] at ($(l1) + (-1.2,-1)$) (l2){$\sericom$};
            \node[draw, circle, very thick,inner sep=\insepseri] at ($(l1) + (0,-1)$) (l3){$\sericom$};
            \node[draw, circle, very thick,inner sep=\insepseri] at ($(l1) + (1.2,-1)$) (l4){$\sericom$};

            \node[draw, circle, very thick,inner sep=\insepe] at ($(l2) + (-0.5,-1)$) (l5){$e_1$};
            \node[draw, circle, very thick,inner sep=\insepe] at ($(l2) + (0,-1)$) (l6){$e_2$};    

            \node[draw, circle, very thick,inner sep=\insepe] at ($(l3) + (-0.5,-1)$) (l7){$e_3$};
            \node[draw, circle, very thick,inner sep=\insepe] at ($(l3) + (0,-1)$) (l8){$e_4$};  
            
            \node[draw, circle, very thick,inner sep=\insepe] at ($(l4) + (-0.5,-1)$) (l9){$e_5$};
            \node[draw, circle, very thick,inner sep=\insepe] at ($(l4) + (0,-1)$) (l10){$e_6$};  
            \node[draw, circle, very thick,inner sep=\insepe] at ($(l4) + (0.5,-1)$) (l11){$e_7$};  
    
            \draw [very thick] (l5) -- (l2) -- (l1) -- (l4) -- (l11);
            \draw [very thick] (l1) -- (l3);
            \draw [very thick] (l2) -- (l6);
            \draw [very thick] (l3) -- (l7);
            \draw [very thick] (l3) -- (l8);
            \draw [very thick] (l4) -- (l9);
            \draw [very thick] (l4) -- (l10);
            \draw [very thick] (l4) -- (l11);
        \end{tikzpicture} 
      \end{minipage} \\
      & & \\
      (a) & (b) & (c)
    \end{tabular}
    \caption{(a) The graph $G$, (b) the decomposition tree of $G$, and (c) the layered decomposition tree of $G$.}
    \label{fig:dec_seri}
  \end{figure}

The recursive definition of a two-terminal series-parallel graph $G$ naturally gives us a rooted full binary tree $T$ representing $G$, called the \emph{decomposition tree} of $G$ (see \Cref{fig:dec_seri}(a) and~(b)).
To avoid confusion, we refer to a vertex and an edge of $T$ as a \emph{node} and a \emph{link}, respectively.
For a node $x$ of $T$, let $T_x$ be a subtree of $T$ rooted at $x$.
Each leaf of $T$ corresponds to an edge of $G$ whose endpoints are labeled with a source $s$ and a sink $t$.
Each internal node $x$ of $T$ is labeled either $\sericom$ or $\paracom$.
Suppose that $x$ has exactly two children $x_1$ and $x_2$. 
The label $\sericom$ indicates a series composition of two-terminal series-parallel graphs defined by $T_{x_1}$ and $T_{x_2}$.
The label $\paracom$ indicates a parallel composition of two-terminal series-parallel graphs defined by $T_{x_1}$ and $T_{x_2}$.
We refer to nodes labeled $\sericom$ as \emph{$\sericom$-nodes} and to nodes labeled $\paracom$ as \emph{$\paracom$-nodes}.
We denote by $G_x$ the graph composed by $T_x$.
Let $r$ be the root of $T$.
Then, we have $G_r = G$.
Note that, since $G_1 \sericom G_2$ and $G_2 \sericom G_1$ produce different two-terminal graphs, we assume that children of a $\sericom$-node are ordered.
In addition, since we have assumed $G$ is $2$-connected, the root $r$ of $T$ is labeled $\paracom$ (assuming $G$ has at least three vertices).


At the beginning of our algorithm, we construct a decomposition tree $T^\prime$ of a given graph $G$ in linear time~\cite{ValdesTL82}, and then transform it into a suitable form for our algorithm as follows (see also \Cref{fig:dec_seri}(c)).
If a $\sericom$-node $x$ of $T^\prime$ has a child $\sericom$-node $x^\prime$, then we contract a link $x x^\prime$ without changing the order of series compositions. 
For example, suppose that $x$ has children $x_1$ and $x^\prime$; $x^\prime$ has children $x_2$ and $x_3$; and $G_x = G_{x_1} \sericom G_{x^\prime}$.
Then, we contract the link $x x^\prime$ so that $G_x = G_1 \sericom G_2 \sericom G_3$.
The contracted tree still tells how to construct $G$.
Similarly, if a $\paracom$-node $x$ of $T^\prime$ has a child $\paracom$-node $x^\prime$, then we contract the link $x x^\prime$.
We iteratively contract such links until each $\sericom$-node has only leaves or $\paracom$-nodes as its children, and $\paracom$-node has only leaves or $\sericom$-nodes as its children.
Note that each $\paracom$-node has at most one leaf of $T^\prime$ as its children because $G$ has no multiple edges.
The tree obtained in this way is called a \emph{layered decomposition tree} of $G$ and is denoted by $T$.

Let $C$ be a cycle of a graph $G = (V,E)$.
For a subgraph $G^\prime = (V^\prime, E^\prime)$ of $G$, we say that $C$ \emph{enters} $G^\prime$ if $C$ has both an edge in $E^\prime$ and an edge in $E \setminus E^\prime$. 
Suppose that there is a $\sericom$-node $x$ of $T$ such that $x$ has $c\ge 4$ children $x_1, x_2, \ldots, x_c$.
Then, no cycle of length at most $4$ enters $G_x$; since $G_x$ is created by series compositions of at least four two-terminal labeled graphs, every cycle entering $G_x$ has length at least $5$.
Thus, the problem can be solved independently in each of $G_x$ and the remaining part.

Assume that each $\sericom$-node $x$ of $T$ has at most three children.
Before explaining dynamic programming over $T$, we give the following key lemma.

\begin{lemma} \label{lem:cycle_traverse}
    Let $x$ be any node of a layered decomposition tree $T$ of a two-terminal series-parallel graph $G$ and $\List$ be a collection of cycles in $G$ of length at most $4$.
    For any solution $\mathcal{S} \subseteq \List$ of $G$, there exists at most one cycle in $\mathcal{S}$ that enters $G_x$.
\end{lemma}

\begin{proof}
    If $x$ is a leaf or the root of $T$, the lemma trivially holds.
    Suppose otherwise, and assume for a contradiction that there exist two edge-disjoint cycles $C, C^\prime \in \mathcal{S}$ that enter $G_x$.

    Suppose that $x$ is a $\sericom$-node.
    Recall that $x$ has two or three children.
    Suppose that $x$ has exactly two children $y$ and $z$, that is, $G_x = G_{y} \sericom G_{z}$.
    We denote by $s_y$ and $t_y$ the source and the sink of $G_{y}$, respectively, and by $s_z$ and $t_z$ analogously.
    Observe that each of $C$ and $C^\prime$ contains at least one edge from $G_y$ and at least one edge from $G_z$.
    Moreover, they have at least one edge outside of $G_x$, implying that they have at most three edges in $G_x$.
    As $C$ contains at most three edges in $G_x$, without loss of generality, we assume that $C$ contains exactly one edge (i.e., $s_yt_y$) of $G_y$.
    As $G$ has no parallel edges and the cycles are edge-disjoint, $C'$ contains exactly two edges of $G_y$.
    Similarly, $C$ contains exactly two edges of $G_z$ and $C'$ contains exactly one edge of $G_z$.
    Thus, both cycles contain exactly three edges of $G_x$.
    This implies that $C$ and $C'$ contain the edge $s_yt_z$ as they have at most four edges, contradicting the fact that $G$ has no parallel edges.
    We can derive a similar contradiction for the case where $x$ has three children.
    
    Next, suppose that $x$ is a $\paracom$-node.
    Let $x^\prime$ be the parent of $x$.
    From the definition of a layered decomposition tree, $x^\prime$ is a $\sericom$-node.
    As $C$ enters $G_x$, $C$ passes through (1) an edge of $G_y$, (2) one of the source or sink of $G_x$, and (3) an edge of $G_x$, where $y$ is a child node of $x'$ with $y \neq x$.
    This implies that $C$ also enters $G_{x'}$ as otherwise $C$ passes through the vertex of (2) twice, contradicting the fact that $C$ is a cycle.
    Also, $C'$ enters $G_{x'}$, which leads to a contradiction as mentioned above. 
\end{proof}

Let $\mathcal{S}$ be a largest solution of $G$.
Suppose that $x$ is a $\sericom$-node of $T$ with $c$ children $x_1, x_2, \ldots, x_c$.
For an integer $i \in \intsec{c}$, let $s_i$ and $t_i$ denote a source and a sink of $G_{x_i}$, respectively.
We distinguish the following two cases to consider: 
\begin{enumerate}
    \item[$(s_1)$] at least one cycle $C$ in $\mathcal{S}$ enters $G_x$ and $s_it_i \in E(C)$ for every integer $i \in \intsec{c}$;
    \item[$(s_2)$] at least one cycle $C$ in $\mathcal{S}$ enters $G_x$ and $s_it_i \notin E(C)$ for some integer $i \in \intsec{c}$.
\end{enumerate}
Similarly, for a $\paracom$-node $x$ with a source $s$ and a sink $t$, we also distinguish the following two cases to consider: 
\begin{enumerate}
    \item[$(p_1)$] at least one cycle $C$ in $\mathcal{S}$ enters $G_x$ and $st \in E(C)$;
    \item[$(p_2)$] at least one cycle $C$ in $\mathcal{S}$ enters $G_x$ and $st \notin E(C)$;
\end{enumerate}
We note that the above cases are not exhaustive: there may be no cycle in $\mathcal{S}$ entering $G_x$.
It is not necessary to consider such a case in the construction of our algorithm.
We also note that by~\Cref{lem:cycle_traverse}, there are no more than one (edge-disjoint) cycle satisfying these conditions.

Let $\List^x$ be a restriction of $\List$ to $G_x$, that is, $\List^x  = \{ H \in \List : E(H) \subseteq E(G_x) \}$. 
In our algorithm, for each node $x$ of $T$, we compute the largest size of a subcollection $\mathcal{S}_x$ with $ \mathcal{S}_x = \mathcal{S} \cap \List^x$.
Let $\seridp{x}$ be the largest size of $\mathcal{S}_x$ for a $\sericom$-node $x$, and let $\paradp{x}$ be the largest size of $\mathcal{S}_x$ for a $\paracom$-node $x$.
Notice that, originally, leaves of $T$ are labeled neither $\sericom$ nor $\paracom$, and hence $\seridp{x}$ and $\paradp{x}$ cannot be defined for the leaves.
For algorithmic simplicity, we consider a leaf $x$ as a $\sericom$-node if its parent is labeled $\paracom$, and as a  $\paracom$-node if its parent is labeled $\sericom$.
This simplification allows us to define $\seridp{x}$ and $\paradp{x}$ for a leaf $x$ accordingly.

We also define the truth values $\seribool{j}{x}$ and $\parabool{j}{x}$ for each $j \in \{1,2\}$ and each node $x$ of $T$.
We set $\seribool{j}{x} = 1$ (resp.\ $\parabool{j}{x} = 1$) if and only if there exists a mutually edge-disjoint subcollection $\mathcal{S}_x^\prime$ of $\List^x$ that satisfies the following conditions: 
\begin{itemize}
    \item $|\mathcal{S}_x^\prime| = \seridp{x}$ (resp.\ $|\mathcal{S}_x^\prime| = \paradp{x}$);
    \item there exists a cycle $C \in \List \setminus \List^x$ corresponding to the case $(s_j)$ (resp.~$(p_j)$);
    \item all subgraphs in $\mathcal{S}_x^\prime$ and $C$ are edge-disjoint.
\end{itemize}
Intuitively speaking, $\seribool{j}{x}=1$ (and $\parabool{j}{x} = 1$) if and only if we can further add a cycle $C$ entering $G_x$ into a partial solution at an
ancestor of $x$.


We are ready to explain how to compute $\seridp{x}$, $\paradp{x}$, $\seribool{j}{x}$, and $\parabool{j}{x}$ for each node $x$ of $T$ and each $j \in \{1,2\}$.

\paragraph{Leaf node.}
Suppose that $x$ is a leaf of $T$.
Let $s$ and $t$ be the source and the sink of $G_x$, respectively.
One can verify that the following equalities hold: $\seridp{x} = \paradp{x} = 0$; $\seribool{1}{x} = \parabool{1}{x} = 1$ if and only if there exists a cycle $C$ in $\List$ such that $st \in E(C)$; and $\seribool{2}{x} = \parabool{2}{x} = 0$.

\paragraph{Internal $\sericom$-node.}
Suppose that $x$ is a $\sericom$-node with $c$ children $x_1, \ldots, x_c$.
Since $G_x$ consists of series compositions of $G_{x_1},\ldots, G_{x_c}$, every cycle in $G_x$ is contained in $G_{x_i}$ for some $i$.
We thus have 
\begin{align*}
    \seridp{x} = \sum_{i \in \intsec{c}} \paradp{x_i}.
\end{align*}

We next compute $\seribool{j}{x}$ for each $j \in \{ 1,2 \}$.
Recall that $x$ has at most three children.

Suppose that $c=3$.
If there exists a cycle $C \in \List \setminus \List^x$ that enters $G_x$, then it passes through $s_1$ and $t_3$, meaning that it enters $G_{x_i}$ for all $i \in \intsec{3}$. 
Conversely, for every $i \in \intsec{3}$, if there is a cycle $C_i \in \List \setminus \List^{x_i}$ such that $C_i$ enters $G_{x_i}$, then it must have $E(C_i) = \{s_1t_1, s_2t_2, s_3t_3,t_3s_1\}$, that is, the cycle is uniquely determined $C=C_i$ for $i \in \intsec{3}$.
It is easy to observe that $C$ is edge-disjoint from any cycles in $\mathcal{S}_x$ if and only if it is edge-disjoint from any cycles in $\mathcal{S}_{x_i}$ for all $i \in \intsec{3}$.
Hence, we have $\seribool{1}{x} = \parabool{1}{x_1} \land \parabool{1}{x_2} \land \parabool{1}{x_3}$.
This also implies that there is no cycle $C \in \mathcal S$ that enters $G_x$ and $s_it_i \notin E(C)$, which yields that $\seribool{2}{x} = 0$. 

Suppose next that $c=2$.
By the similar argument to the case $c = 3$, we have $\seribool{1}{x} = 1$ if and only if $\parabool{1}{x_1} \land \parabool{1}{x_2} = 1$ and there is a cycle $C \in \List \setminus \List^x$ such that $s_1t_1, s_2t_2 \in E(C)$.
We explain how to decide $\seribool{2}{x}$.
If there is a cycle $C \in \List \setminus \List^x$ with $s_1t_1 \notin E(C)$ that enters $G_x$, then it enters both $G_{x_1}$ and $G_{x_2}$, and it holds that $s_2t_2 \in E(C)$ because the length of $C$ is at most $4$.
Conversely, for a cycle $C_1 \in \List \setminus \List^{x_1}$ such that $C_1$ enters $G_{x_1}$ and $s_1t_1 \notin E(C_1)$, $C_1$ also enters $G_{x_2}$ and $G_x$, and $s_2t_2 \in E(C_1)$ holds.
The same argument is applied to a cycle $C \in \List \setminus \List^x$ with $s_2t_2 \notin E(C)$ that enters $G_x$.
Thus, we have $\seribool{2}{x} = (\parabool{2}{x_1} \land \parabool{1}{x_2}) \lor (\parabool{1}{x_1} \land \parabool{2}{x_2})$.

\paragraph{Internal $\paracom$-node.}
Suppose that $x$ is a $\paracom$-node with $c$ children $x_1, \ldots, x_c$.
Let $s$ and $t$ be the source and the sink of $G_x$, respectively.

To compute $\paradp{x}$, we construct an \emph{auxiliary graph} $A_x$ whose vertex set is $\{a_1, a_2, \ldots, a_c \}$.
We associate each child $x_i$ of $x$ with a vertex $a_i$.
Let $i, j \in \intsec{c}$ be distinct integers.
Suppose that $x_i$ and $x_j$ are internal nodes of $T$.
Then, $A_x$ has an edge $a_ia_j$ if $\seribool{1}{x_i} \land \seribool{1}{x_j} = 1$ and there exists a cycle in $\List^x$ that enters both $G_{x_i}$ and $G_{x_j}$.
Note that such a cycle $C$ satisfies $|E(C) \cap E(G_{x_i})| = |E(C) \cap E(G_{x_j})|=2$, which means that $C$ must satisfy the case~$(s_1)$ for $\sericom$-nodes $x_i$ and $x_j$.
Suppose next that $x_i$ is an internal node and $x_j$ is a leaf of $T$.
In this case, $st \in E(G_{x_j})$.
Then, $A_x$ has an edge $a_ia_j$ if at least one of the following conditions is satisfied:
\begin{enumerate}
    \item $x_i$ has exactly $c$ children with $c \in \{2,3\}$, $\seribool{1}{x_i} = 1$, and there exists a cycle $C$ in $\List^x$ of length $c+1$ that enters both $G_{x_i}$ and $G_{x_j}$; or
    \item $\seribool{2}{x_i} = 1$.
\end{enumerate}
Note that, in the second case $\seribool{2}{x_i} = 1$, there is a cycle $C \in \List \setminus \List^{x_i}$ entering $G_{x_i}$ such that $x_i$ has a child $y$ with $|E(C) \cap E(G_y)| \ge 2$.
This implies that $C$ has exactly three edges in $G_{x_i}$ and hence we have $st \in E(C)$.
Also note that there is no case that both $x_i$ and $x_j$ are leaves because $G$ has no parallel edges.
We complete the construction of $A_x$.

The intuition of the auxiliary graph $A_x$ is as follows.
If there is an edge $a_ia_j \in A_x$, then we can further add a cycle $C$ in $G_{x_i} \paracom G_{x_j}$ that is edge-disjoint from any cycles in $\bigcup_{h \in \intsec{c}}\mathcal S_{x_h}$.
We can simultaneously add such cycles for other edges in $A_x$.
However, by~\Cref{lem:cycle_traverse}, we cannot add more than one cycles entering $G_{x_i}$.
Thus, in order to add as many such cycles as possible, the corresponding edges must form a matching in $A_x$.
In fact, the following equality holds.
\begin{align} \label{eq:matching_dp}
    \paradp{x} = \sum_{i \in \intsec{c}} \seridp{x_i} + |M_x^\ast|,
\end{align}
where $M_x^\ast$ be a maximum matching of $A_x$.
The correctness of \Cref{eq:matching_dp} will be given in \Cref{subsec:correctness_dp}.

To compute $\parabool{j}{x}$ for each $j \in \{ 1,2 \}$, we construct additional auxiliary graphs $A_x^1$ and $A_x^2$ from $A_x$.
Let $A_x^1$ be the graph obtained from $A_x$ as follows.
We first add a vertex $a^{\prime}$.
Then, we add an edge $a^\prime a_i$ if $x$ has a leaf child $x_i$ and there is a cycle $C \in \List \setminus \List^x$ such that $C$ enters $G_{x}$ and $st \in E(C)$.

Similarly, let $A_x^2$ be the graph obtained from $A_x$ as follows.
We first add a vertex~$a^{\prime\prime}$.
Then, for each $i \in \intsec{c}$, we add an edge $a^{\prime\prime}a_i$ if $x_i$ is an internal node of $T$, $\seribool{1}{x_i} = 1$ and there is a cycle $C \in \List \setminus \List^x$ that enters $G_{x_i}$.

Let $M_x^1$ and $M_x^2$ be maximum matchings of $A_x^1$ and $A_x^2$, respectively.
We let $\parabool{1}{x} = 1$ if and only if $|M_x^1| > |M_x^\ast|$; and $\parabool{2}{x} = 1$ if and only if $|M_x^2| > |M_x^\ast|$.
The correctness of the computation of these truth values will be given in \Cref{subsec:correctness_dp}.

Finally, we conclude that $\paradp{r}$ is the size of a largest solution of $G$.

\subsubsection{Correctness} \label{subsec:correctness_dp}

\begin{lemma} \label{lem:matching_dp}
    \Cref{eq:matching_dp} is correct.
\end{lemma}

\begin{proof}
    We first show that $\paradp{x} \ge \sum_{i \in \intsec{c}} \seridp{x_i} + |M_x^\ast|$.
    For $i \in \intsec{c}$, let $\mathcal{S}_{x_i}$ be a mutually edge-disjoint subcollection of $\List^{x_i}$ such that $|\mathcal{S}_{x_i}| = \seridp{x_i}$.
    Recall that each edge $a_i a_j$ in $M_x^\ast$ corresponds to a cycle in $\List^x$ that enters both $G_{x_i}$ and $G_{x_j}$ and edge-disjoint from any cycle in $\mathcal{S}_{x_i} \cup \mathcal{S}_{x_j}$.
    Let $\mathcal{M}_x$ be a subcollection of cycles in $\List^x$ corresponding to $M_x^\ast$.
    Observe that cycles in $\mathcal{M}_x$ are mutually edge-disjoint because $M_x^\ast$ is a matching of $A_x$.
    Therefore, $(\bigcup_{i \in \intsec{c}} \mathcal{S}_{x_i}) \cup \mathcal{M}_x$ is mutually edge-disjoint, and hence we have $\paradp{x} \ge \sum_{i \in \intsec{c}} \seridp{x_i} + |M_x^\ast|$.

    We next show that $\paradp{x} \le \sum_{i \in \intsec{c}} \seridp{x_i} + |M_x^\ast|$.
    Suppose that $\mathcal{S}_{x}$ is a mutually edge-disjoint subcollection of $\List^x$ such that $|\mathcal{S}_{x}| = \paradp{x}$.
    Consider a cycle $C$ in $\mathcal{S}_{x}$ entering both $G_{x_i}$ and $G_{x_j}$ such that $x_i$ and $x_j$ are internal nodes of $T$.
    If $a_i a_j \notin E(A_x)$, then it holds that $\seribool{1}{x_i} \land \seribool{1}{x_j} = 0$.
    Without loss of generality, assume that $\seribool{1}{x_i} = 0$.
    Let $\mathcal{S}_{x_i} = \mathcal{S}_{x} \cap \List^{x_i}$.
    Observe that $|\mathcal{S}_{x_i}| -1 \le \seridp{x_i}$ as otherwise, the fact that the cycle $C$ is edge-disjoint from any cycles in $\mathcal{S}_{x_i}$ implies that $\seribool{1}{x_i} = 1$.
    Then, we can replace $\mathcal{S}_{x_i} \cup \{C\}$ with an edge-disjoint subcollection $\mathcal{S}_{x_i}^\prime$ with $|\mathcal{S}_{x_i}^\prime| = \seridp{x_i}$ in $\mathcal{S}_{x}$ (i.e., $\mathcal{S}_{x} \coloneqq (\mathcal{S}_{x} \setminus (\mathcal{S}_{x_i} \cup \{C\}) \cup \mathcal{S}_{x_i}^\prime$).
    The same argument is applicable to the case where one of $x_i$ and $x_j$ is a leaf of $T$.
    Applying the above replacement exhaustively, we eventually obtain a mutually edge-disjoint subcollection $\mathcal{S}_{x}^\ast \subseteq \List^x$ with $|\mathcal{S}_{x}^\ast| = \paradp{x}$ that satisfied the following condition: For each cycle $C$ in $\mathcal{S}_{x}^\ast$ that enters both $G_{x_i}$ and $G_{x_j}$ for  distinct $i, j \in \intsec{c}$, $A_x$ has an edge $a_i a_j$.
    We denote by $\mathcal{C} \subseteq \mathcal{S}_{x}^\ast$ the set of cycles, each of which enters $G_{x_i}$ and $G_{x_j}$ for distinct $i, j \in \intsec{c}$.
    Let $M$ be the subset of $E(A_x)$ corresponding to $\mathcal{C}$.
    Then, $M$ forms a matching of $A_x$; otherwise, there are two edge-disjoint cycles that enter $G_{x_i}$ for some $i \in \intsec{c}$, which contradicts Lemma~\ref{lem:cycle_traverse}.
    Denote $\mathcal{S}_{x_i}^\ast = \mathcal{S}_{x}^\ast \cap \List^{x_i}$ for each $i \in \intsec{c}$.
    Obviously, it holds that $|\mathcal{S}_{x_i}^\ast| \le \seridp{x_i}$.
    Since $\mathcal{S}_{x}^\ast = (\bigcup_{i \in \intsec{c}} \mathcal{S}_{x_i}^\ast) \cup \mathcal{C}$, we thus have $\paradp{x} \le \sum_{i \in \intsec{c}} \seridp{x_i} + |M| \le \sum_{i \in \intsec{c}} \seridp{x_i} + |M_x^\ast|$ as $M_x^\ast$ is a maximum matching of $A_x$.
    This completes the proof of Lemma~\ref{lem:matching_dp}. 
\end{proof}

\begin{lemma} \label{lem:boolean_dp}
    Let $M_x^1$ and $M_x^2$ denote maximum matchings of $A_x^1$ and $A_x^2$, respectively.
    Then, $\parabool{1}{x} = 1$ if and only if $|M_x^1| > |M_x^\ast|$; and $\parabool{2}{x} = 1$ if and only if $|M_x^2| > |M_x^\ast|$.
\end{lemma}

\begin{proof}
    We prove the former claim of the lemma.
    Suppose that $\parabool{1}{x} = 1$.
    From the definition of $\parabool{1}{x}$, there exists a mutually edge-disjoint subcollection $\mathcal{S}_x^\prime$ of $\List^x$ with $|\mathcal{S}_x^\prime| = \paradp{x}$ and a cycle $C \in \List \setminus \List^x$ with $st\in E(C)$ that is edge-disjoint from any cycle in $\mathcal{S}_x^\prime$.
    Then, we can construct a matching $M$ of $A_x$ corresponding to a maximum subcollection $\mathcal{C} \subseteq \mathcal{S}_x^\prime$ such that each cycle in $\mathcal{C}$ enters $G_{x_i}$ and $G_{x_j}$ for some $i, j \in \intsec{c}$.
    By~\Cref{lem:cycle_traverse}, such a subcollection $\mathcal C$ is uniquely determined.
    As $M_x^\ast$ is a maximum matching of $A_x$, we have $|M| \le |M_x^\ast|$.
    We now claim that $|M| = |M_x^\ast|$.
    Suppose to the contrary that $|M| < |M_x^\ast|$.
    Denote $\mathcal{S}_{x_i}^\prime = \mathcal{S}_{x}^\prime \cap \List^{x_i}$ for each $i \in \intsec{c}$.
    Since $\mathcal{S}_x^\prime = (\bigcup_{i \in \intsec{c}} \mathcal{S}_{x_i}^\prime) \cup \mathcal{C}$, we have $|\mathcal{S}_x^\prime| = (\sum_{i \in \intsec{c}} |\mathcal{S}_{x_i}^\prime|) + |M|$.
    Combined with the facts $\paradp{x} = \sum_{i \in \intsec{c}} \seridp{x_i} + |M_x^\ast|$, $|M| < |M_x^\ast|$, and $|\mathcal{S}_x^\prime| = \paradp{x}$, we can observe that there exists $p \in \intsec{c}$ such that $|\mathcal{S}_{x_p}^\prime| > \seridp{x_p}$, which contradicts the definition of $\seridp{x_p}$.
    Thus, we have $|M| = |M_x^\ast|$.
    
    Let $x_i$ be the child of $x$ corresponding to the edge $st$.
    Since $st\in E(C)$ and $C$ is edge-disjoint from any cycles in $\mathcal{S}_x^\prime$, every edge in $M$ is not incident to $a_i$.
    Thus, $M \cup \{a^\prime a_i\}$ forms a matching of $A_x^1$.
    As $|M| = |M_x^\ast|$, we have $|M_x^1| \ge |M \cup \{a^\prime a_p\}| > |M| = |M_x^\ast|$.
    
    Conversely, suppose that $|M_x^1| > |M_x^\ast|$.
    This implies that there is an edge $e \in M_x^1$ incident to $a^\prime$.
    Let $C$ be the cycle in $\List \setminus \List^x$ corresponding to $e$.
    As in the proof of Lemma~\ref{lem:matching_dp}, we can construct a mutually edge-disjoint subcollection $\mathcal{S}^\prime_x$ of $\List^x$ such that $|\mathcal{S}^\prime_x| = \sum_{i \in \intsec{c}} \seridp{x_i} + |M_x^1 \setminus \{e\}|$. Since $|M_x^1 \setminus \{e\}| = |M_x^1|-1 \ge |M_x^\ast|$, we have $|\mathcal{S}^\prime_x| \ge \sum_{i \in \intsec{c}} \seridp{x_i} + |M_x^\ast| = \paradp{x}$.
    It clearly holds that $|\mathcal{S}^\prime_x| \le \paradp{x}$, and hence we have $|\mathcal{S}^\prime_x| = \paradp{x}$.
    Moreover, every subgraph in $\mathcal{S}^\prime_x$ and $C$ are edge-disjoint.
    Therefore, we have $\parabool{1}{x} = 1$.

    The latter claim can be proved in the similar way, which completes the proof of Lemma~\ref{lem:boolean_dp}.
\end{proof}

\subsubsection{Running time}
Let $n$ and $m$ be the number of vertices and edges of a given series-parallel graph $G$, respectively, and let $\List$ be a collection of cycles in $G$ of length at most $4$.
As noted at the beginning of Section~\ref{subsec:sp_4cycle}, our algorithm is applied to each of $\bicon^\ast$ biconnected components $G_1,G_2,\ldots, G_{\bicon^\ast}$ of $G$.
We can enumerate biconnected components of $G$ in $O(n+m)$ time~\cite{HopcroftT73} and partition $\List$ into $\List^1, \List^2, \ldots, \List^{\bicon^\ast}$ so that $\List^\bicon$ is the subcollection of $\List$ consisting of all cycles in $G_\bicon$ for each $\bicon \in \intsec{\bicon^\ast}$ in $O(|\List|)$ time.

For $\bicon \in \intsec{\bicon^\ast}$, we denote by $n_\bicon$ and $m_\bicon$ the number of vertices and edges of $G_\bicon$, respectively.
We may assume that edges of $G_\bicon$ are labeled distinct integers $1,2,\ldots,m_\bicon$.
Each cycle $C$ of $G_\bicon$ with $h$ edges produces a word $x_1x_2\dots x_h$ such that $x_i \in \intsec{m_\bicon}$ for every $i \in \intsec{h}$ and $x_1 < x_2 < \cdots < x_h$.
As a preprocessing, we convert every cycle in $\List^\bicon$ to a corresponding word, and then lexicographically sort $\List^\bicon$ with radix sort.
Since every cycle in $\List^\bicon$ is of length at most $4$, this can be done in $O(|\List^\bicon| + m_\bicon )$ time.
Using this data structure, we can check whether given a cycle $C$ is in $\List$ in time $O(\log |\List^\bicon|)$ with binary search.

After the preprocessing, we construct an (original) decomposition tree of $G_\bicon$ in $O(n_\bicon)$ time~\cite{ValdesTL82}.
It is not hard to see that the decomposition tree can be modified into a layered decomposition tree $T_\bicon$ in $O(n_\bicon)$ time with depth-first search.
Moreover, we record the source and the sink of $G_x$ in each node $x$ of $T_\bicon$.

We bound the running time of our dynamic programming.
Obviously, for each leaf $x$ of $T_\bicon$, we can compute $\seridp{x}$, $\paradp{x}$, $\seribool{j}{x}$, and $\parabool{j}{x}$ for each $j \in \{1,2\}$ in $O(1)$ time.
Moreover, $\seridp{x}$, $\seribool{1}{x}$, and $\seribool{2}{x}$ are computed in $O(1)$ time for each $\sericom$-node $x$ of $T_\bicon$.

Consider the case where $x$ is an internal $\paracom$-node of $T_\bicon$.
Let $c_x$ denote the number of children of $x$.
Our algorithm first constructs the auxiliary graph $A_x$ with the vertex set $\{a_1, \ldots, a_{c_x} \}$.
Let $i, j \in \intsec{c_x}$ be distinct integers.
Suppose that $x_i$ and $x_j$ are internal nodes of $T_\bicon$.
We check whether $\seribool{1}{x_i} \land \seribool{1}{x_j} = 1$ in $O(1)$ time.
Moreover, we check whether there exists a cycle $C$ that enters both $G_{x_i}$ and $G_{x_j}$.
In fact, the cycle $C$ is uniquely determined if it exists: each of $G_{x_i}$ and $G_{x_j}$ is obtained by a series composition of two graphs, meaning that $|E(C) \cap E(G_{x_i})| = |E(C) \cap E(G_{x_j})|=2$.
We can compute such a cycle $C$ in $O(1)$ time by recording the source and the sink of a graph corresponding to each node, while we can decide whether $C \in \List^\bicon$ in $O(\log |\List^\bicon|)$ time using the above data structure.
As $|\List^\bicon|$ is bounded by $n_\bicon^4$ above, this can be done in $O(\log n_\bicon)$ time, which also decides whether $a_i a_j \in E(A_x)$ or not.
Similarly, for the case where $x_i$ is an internal node and $x_j$ is a leaf of $T_\bicon$, we can determine in $O(\log n_\bicon)$ time whether $a_i a_j \in E(A_x)$ or not.
Therefore, the construction of $A_x$ takes $O(c_x^2 \log n_\bicon)$ time.

We then construct graphs $A_x^1$ and $A_x^2$ from $A_x$ to compute $\parabool{1}{x}$ and $\parabool{2}{x}$.
For the graph $A_x^1$, we need to decide whether $a^\prime a_i \in E(A_x^1)$, where $x_i$ for $i \in \intsec{c_x}$ is the unique leaf child of $x$ (if it exists).
To this end, we construct the collection $\mathcal{C}_x^1$ of all cycles that enter $G_x$ and contain the edge of $G_{x_i}$.
If $x$ is the root of $T$, then clearly $\mathcal{C}_x^1 = \emptyset$.
Otherwise $x$ has the parent $y$ labeled $\sericom$ and $y$ has the parent $z$ labeled $\paracom$ as the root of $T_\bicon$ has label $\paracom$.
If $y$ has three children, $\mathcal{C}_x^1$ can be obtained in $O(1)$ time because a possible cycle contained in $\mathcal{C}_x^1$ is uniquely determined.
Suppose that $y$ has exactly two children and let $x^\prime$ be a child of $y$ with $x^\prime \neq x$.
For any cycle $C \in \mathcal{C}_x^1$, the following two cases are considered: (i) $C$ shares exactly one edge with $G_{x^\prime}$; and (ii) $C$ shares exactly two edges with $G_{x^\prime}$.
In the case (i), for each child $c_z$ of $z$, at most two edges in $G_{c_z}$ that can be shared with $C$ are uniquely determined because $c_z$ is a leaf of $T_\bicon$ or labeled~$\sericom$. 
In the case (ii), for each child $c_y$ of $y$, exactly two edges in $G_{c_y}$ that can be shared with $C$ are uniquely determined and $C$ contains an edge between the source and the sink of $G_z$.
In both cases, $\mathcal{C}_x^1$ of size $O(n_\bicon)$ can be constructed in $O(n_\bicon)$ time.
We thus decide in $O(n_\bicon \log n_\bicon)$ time whether there is a cycle $C \in \mathcal{C}_x^1 \cap \List^\bicon$, that is, $a^\prime a_i \in E(A_x^1)$.
For the graph $A_x^2$, we decide whether $a^{\prime\prime} a_i \in E(A_x^1)$ for each $i \in \intsec{c_x}$ such that $x_i$ is an internal node of $T_\bicon$.
We check whether $\seribool{1}{x_i} = 1$, and if so, there exists a cycle $C \in \List \setminus \List^x$ that enters $G_{x_i}$.
Since $G_{x_i}$ shares exactly two edges with $C$, the cycle is uniquely determined for each $i \in \intsec{c_x}$ as in the case (ii) above.
Thus, $A_x^2$ is constructed in $O(c_x \log n_\bicon)$ time.
After the construction of $A_x$, $A_x^1$, and $A_x^2$, we obtain maximum matchings $M_x^\ast$, $M_x^1$, and $M_x^2$ in $O(c_x^{2.5})$ time, respectively~\cite{MicaliV80}.
Therefore, $\paradp{x}$ is computed in $O(c_x^2 \log n_\bicon + n_\bicon \log n_\bicon + c_x^{2.5})$ time for each $\paracom$-node $x$ of $T$. 

In summary, $\paradp{r}$ is obtained in time 
\[
O(m_\bicon + |\List^\bicon| + \sum_{x \in V(T)} (c_x^2 \log n_\bicon + n_\bicon \log n_\bicon + c_x^{2.5})).
\]
Recall that $\sum_{x \in V(T)} c_x = O(n_\bicon)$ and $m_\bicon = O(n_\bicon)$ hold. 
Therefore, our algorithm for $G_\bicon$ runs in $O(|\List^\bicon| + n_\bicon^{2.5})$ time. 
Since $\sum_{\bicon \in \intsec{\bicon^\ast}} n_\bicon = n + b^\ast \le 2n$, we conclude that the total running time for a given series-parallel graph $G$ is bounded by $O(|\List| + n^{2.5} )$.
This completes the proof of Theorem~\ref{thm:sp_4cycle}.

\section*{Acknowledgements}
We thank the referees for their valuable comments and suggestions which greatly helped to improve the presentation of this paper.

%
%
%
\begin{sloppypar}
\printbibliography

@article{CapraraR02:IPL:Packing,
  author    = {Alberto Caprara and
               Romeo Rizzi},
  title     = {Packing triangles in bounded degree graphs},
  journal   = {Inf. Process. Lett.},
  volume    = {84},
  number    = {4},
  pages     = {175--180},
  year      = {2002},
  url       = {https://doi.org/10.1016/S0020-0190(02)00274-0},
  doi       = {10.1016/S0020-0190(02)00274-0},
  bibsource = {dblp computer science bibliography, https://dblp.org}
}

@article{Minty80:JCTB:maximal,
  author    = {George J. Minty},
  title     = {On maximal independent sets of vertices in claw-free graphs},
  journal   = {J. Comb. Theory, Ser. {B}},
  volume    = {28},
  number    = {3},
  pages     = {284--304},
  year      = {1980},
  url       = {https://doi.org/10.1016/0095-8956(80)90074-X},
  doi       = {10.1016/0095-8956(80)90074-X},
  bibsource = {dblp computer science bibliography, https://dblp.org}
}

@article{Sbihi80:DM:Algorithme,
  author    = {Najiba Sbihi},
  title     = {Algorithme de recherche d'un stable de cardinalite maximum dans un
               graphe sans etoile},
  journal   = {Discret. Math.},
  volume    = {29},
  number    = {1},
  pages     = {53--76},
  year      = {1980},
  url       = {https://doi.org/10.1016/0012-365X(90)90287-R},
  doi       = {10.1016/0012-365X(90)90287-R},
  bibsource = {dblp computer science bibliography, https://dblp.org}
}

@article{CorneilMH94:DAM:Edge-disjoint,
  author       = {Derek G. Corneil and
                  Shigeru Masuyama and
                  S. Louis Hakimi},
  title        = {Edge-disjoint packings of graphs},
  journal      = {Discret. Appl. Math.},
  volume       = {50},
  number       = {2},
  pages        = {135--148},
  year         = {1994},
  url          = {https://doi.org/10.1016/0166-218X(92)00153-D},
  doi          = {10.1016/0166-218X(92)00153-D},
  bibsource    = {dblp computer science bibliography, https://dblp.org}
}

@article{KirkpatrickH83:SICOMP:complexity,
  author       = {David G. Kirkpatrick and
                  Pavol Hell},
  title        = {On the Complexity of General Graph Factor Problems},
  journal      = {{SIAM} J. Comput.},
  volume       = {12},
  number       = {3},
  pages        = {601--609},
  year         = {1983},
  url          = {https://doi.org/10.1137/0212040},
  doi          = {10.1137/0212040},
  bibsource    = {dblp computer science bibliography, https://dblp.org}
}

@article{Murphy92:DAM:Computing,
  author       = {Owen J. Murphy},
  title        = {Computing independent sets in graphs with large girth},
  journal      = {Discret. Appl. Math.},
  volume       = {35},
  number       = {2},
  pages        = {167--170},
  year         = {1992},
  url          = {https://doi.org/10.1016/0166-218X(92)90041-8},
  doi          = {10.1016/0166-218X(92)90041-8},
  bibsource    = {dblp computer science bibliography, https://dblp.org}
}

@article{Bodlaender98,
title = {A partial $k$-arboretum of graphs with bounded treewidth},
journal = {Theoretical Computer Science},
volume = {209},
number = {1},
pages = {1--45},
year = {1998},
issn = {0304-3975},
doi = {10.1016/S0304-3975(97)00228-4},
author = {Hans L. Bodlaender},
}

@InProceedings{XuZ18,
author="Xu, Chenyang and Zhang, Guochuan",
editor="Wang, Lusheng and Zhu, Daming",
title="The Path Set Packing Problem",
booktitle="the 24th International Conference {(COCOON 2018)}",
year="2018",
publisher="Springer International Publishing",
address="Cham",
doi          = {10.1007/978-3-319-94776-1_26},
pages="305--315",
}

@InProceedings{AravindS23,
author="Aravind, N. R. and Saxena, Roopam",
editor="Lin, Chun-Cheng and Lin, Bertrand M. T. and Liotta, Giuseppe",
title="Parameterized Complexity of Path Set Packing",
booktitle="the 17th International Conference and Workshops {(WALCOM 2023)}",
year="2023",
publisher="Springer Nature Switzerland",
address="Cham",
pages="291--302",
doi          = {10.1007/978-3-031-27051-2_25},
}

@InProceedings{MicaliV80,
  author={Micali, Silvio and Vazirani, Vijay V.},
  booktitle={21st Annual Symposium on Foundations of Computer Science (sfcs 1980)}, 
  title={An $O(\sqrt{|V|} |E|)$ algorithm for finding maximum matching in general graphs}, 
  year={1980},
  volume={},
  number={},
  pages={17--27},
  doi={10.1109/SFCS.1980.12}
}

@article{HopcroftT73,
author = {Hopcroft, John and Tarjan, Robert},
title = {Algorithm 447: Efficient Algorithms for Graph Manipulation},
year = {1973},
publisher = {Association for Computing Machinery},
address = {New York, NY, USA},
volume = {16},
number = {6},
issn = {0001-0782},
url = {https://doi.org/10.1145/362248.362272},
doi = {10.1145/362248.362272},
journal = {Communications of the ACM},
pages = {372-–378},
}

@article{ChungG81,
  title={Recent results in graph decompositions},
  author={Chung, F.R.K and Graham, R.L.},
  journal={London Mathematical Society, Lecture Note Series},
  volume={52},
  pages={103--123},
  year={1981}
}

@article{DyerF85,
author = {M.E. Dyer and A.M. Frieze},
title = {On the complexity of partitioning graphs into connected subgraphs},
journal = {Discrete Applied Mathematics},
volume = {10},
number = {2},
pages = {139--153},
year = {1985},
issn = {0166-218X},
doi = {10.1016/0166-218X(85)90008-3},
url = {https://www.sciencedirect.com/science/article/pii/0166218X85900083},
}

@article{BermanJLSS90,
author = {Fran Berman and David Johnson and Tom Leighton and Peter W. Shor and Larry Snyder},
title = {Generalized planar matching},
journal = {Journal of Algorithms},
volume = {11},
number = {2},
pages = {153--184},
year = {1990},
issn = {0196-6774},
doi = {10.1016/0196-6774(90)90001-U},
url = {https://www.sciencedirect.com/science/article/pii/019667749090001U},
}

@inproceedings{KosowskiMZ05,
  author       = {Adrian Kosowski and
                  Micha{\l} Ma{\l}afiejski and
                  Pawe{\l} {\.{Z}}yli{\'{n}}ski},
  editor       = {Roman Wyrzykowski and
                  Jack J. Dongarra and
                  Norbert Meyer and
                  Jerzy Wasniewski},
  title        = {Parallel Processing Subsystems with Redundancy in a Distributed Environment},
  booktitle    = {Parallel Processing and Applied Mathematics, 6th International Conference,
                  {PPAM} 2005},
  series       = {Lecture Notes in Computer Science},
  volume       = {3911},
  pages        = {1002--1009},
  publisher    = {Springer},
  year         = {2005},
  url          = {https://doi.org/10.1007/11752578_121},
  doi          = {10.1007/11752578_121},
}

@book{CyganFKLMPPS15,
  author       = {Marek Cygan and
                  Fedor V. Fomin and
                  {\L}ukasz Kowalik and
                  Daniel Lokshtanov and
                  D{\'{a}}niel Marx and
                  Marcin Pilipczuk and
                  Micha{\l} Pilipczuk and
                  Saket Saurabh},
  title        = {Parameterized Algorithms},
  publisher    = {Springer},
  year         = {2015},
  url          = {https://doi.org/10.1007/978-3-319-21275-3},
  doi          = {10.1007/978-3-319-21275-3},
  isbn         = {978-3-319-21274-6},
}

@article{FellowsHRV09,
  author       = {Michael R. Fellows and
                  Danny Hermelin and
                  Frances A. Rosamond and
                  St{\'{e}}phane Vialette},
  title        = {On the parameterized complexity of multiple-interval graph problems},
  journal      = {Theor. Comput. Sci.},
  volume       = {410},
  number       = {1},
  pages        = {53--61},
  year         = {2009},
  url          = {https://doi.org/10.1016/j.tcs.2008.09.065},
  doi          = {10.1016/j.tcs.2008.09.065},
  bibsource    = {dblp computer science bibliography, https://dblp.org}
}

@article{ValdesTL82,
author = {Valdes, Jacobo and Tarjan, Robert E. and Lawler, Eugene L.},
title = {The Recognition of Series Parallel Digraphs},
journal = {SIAM Journal on Computing},
volume = {11},
number = {2},
pages = {298--313},
year = {1982},
doi = {10.1137/0211023},
}

@Inproceedings{BafnaP93,
  author={Bafna, V. and Pevzner, P.A.},
  booktitle={Proceedings of 1993 IEEE 34th Annual Foundations of Computer Science}, 
  title={Genome rearrangements and sorting by reversals}, 
  year={1993},
  volume={},
  number={},
  pages={148--157},
  doi={10.1109/SFCS.1993.366872}
}

@article{BontridderHHHLRS03,
  author       = {Koen M. J. De Bontridder and
                  Bjarni V. Halld{\'{o}}rsson and
                  Magn{\'{u}}s M. Halld{\'{o}}rsson and
                  Cor A. J. Hurkens and
                  Jan Karel Lenstra and
                  R. Ravi and
                  Leen Stougie},
  title        = {Approximation algorithms for the test cover problem},
  journal      = {Math. Program.},
  volume       = {98},
  number       = {1-3},
  pages        = {477--491},
  year         = {2003},
  url          = {https://doi.org/10.1007/s10107-003-0414-6},
  doi          = {10.1007/s10107-003-0414-6},
}

@InProceedings{MalafiejskiZ05,
author="Ma{\l}afiejski, Micha{\l}
and {\.{Z}}yli{\'{n}}ski, Pawe{\l}",
editor="Gervasi, Osvaldo
and Gavrilova, Marina L.
and Kumar, Vipin
and Lagan{\`a}, Antonio
and Lee, Heow Pueh
and Mun, Youngsong
and Taniar, David
and Tan, Chih Jeng Kenneth",
title="Weakly Cooperative Guards in Grids",
booktitle="Computational Science and Its Applications -- ICCSA 2005",
year="2005",
publisher="Springer Berlin Heidelberg",
address="Berlin, Heidelberg",
pages="647--656",
isbn="978-3-540-32043-2",
doi={10.1007/11424758_68},
}

@article{ErdosP65,
author={Erd\H{o}s, P. and P\'{o}sa, L.},
title={On Independent Circuits Contained in a Graph},
volume={17},
DOI={10.4153/CJM-1965-035-8},
journal={Canadian Journal of Mathematics},
publisher={Cambridge University Press},
year={1965},
pages={347–-352}
}

@article{LokshtanovMSZ19,
  author       = {Daniel Lokshtanov and
                  Amer E. Mouawad and
                  Saket Saurabh and
                  Meirav Zehavi},
  title        = {Packing Cycles Faster Than Erdos-Posa},
  journal      = {{SIAM} J. Discret. Math.},
  volume       = {33},
  number       = {3},
  pages        = {1194--1215},
  year         = {2019},
  url          = {https://doi.org/10.1137/17M1150037},
  doi          = {10.1137/17M1150037},
}

@article{Bodlaender94,
  author       = {Hans L. Bodlaender},
  title        = {On Disjoint Cycles},
  journal      = {Int. J. Found. Comput. Sci.},
  volume       = {5},
  number       = {1},
  pages        = {59--68},
  year         = {1994},
  url          = {https://doi.org/10.1142/S0129054194000049},
  doi          = {10.1142/S0129054194000049},
}

@article{JiangXZ16,
author = {Minghui Jiang and Ge Xia and Yong Zhang},
title = {Edge-disjoint packing of stars and cycles},
journal = {Theoretical Computer Science},
volume = {640},
pages = {61--69},
year = {2016},
issn = {0304-3975},
doi = {https://doi.org/10.1016/j.tcs.2016.06.001},
url = {https://www.sciencedirect.com/science/article/pii/S0304397516302377},
}

@article{HeathV98,
  author       = {Lenwood S. Heath and
                  John Paul C. Vergara},
  title        = {Edge-Packing in Planar Graphs},
  journal      = {Theory Comput. Syst.},
  volume       = {31},
  number       = {6},
  pages        = {629--662},
  year         = {1998},
  url          = {https://doi.org/10.1007/s002240000107},
  doi          = {10.1007/s002240000107},
}

@inproceedings{GuruswamiRCCW98,
  author       = {Venkatesan Guruswami and
						C. Pandu Rangan and
						Maw-Shang Chang and
						Gerard J. Chang and
						C. K. Wong},
	editor      = {Juraj Hromkovic and
	Ondrej S{\'{y}}kora},
	title        = {The Vertex-Disjoint Triangles Problem},
	booktitle    = {Graph-Theoretic Concepts in Computer Science, 24th International Workshop, {WG} '98},
	series       = {Lecture Notes in Computer Science},
	volume       = {1517},
	pages        = {26--37},
	publisher    = {Springer},
	year         = {1998},
	doi        = {10.1007/10692760\_3},
}

@article{AkiyamaC90,
  author       = {Jin Akiyama and
                  Vasek Chv{\'{a}}tal},
  title        = {Packing paths perfectly},
  journal      = {Discret. Math.},
  volume       = {85},
  number       = {3},
  pages        = {247--255},
  year         = {1990},
  url          = {https://doi.org/10.1016/0012-365X(90)90382-R},
  doi          = {10.1016/0012-365X(90)90382-R},
}

@article{MasuyamaI91,
  author       = {Shigeru Masuyama and
                  Toshihide Ibaraki},
  title        = {Chain Packing in Graphs},
  journal      = {Algorithmica},
  volume       = {6},
  number       = {6},
  pages        = {826--839},
  year         = {1991},
  url          = {https://doi.org/10.1007/BF01759074},
  doi          = {10.1007/BF01759074},
}

@article{FernauR09,
  author       = {Henning Fernau and
                  Daniel Raible},
  title        = {A parameterized perspective on packing paths of length two},
  journal      = {J. Comb. Optim.},
  volume       = {18},
  number       = {4},
  pages        = {319--341},
  year         = {2009},
  url          = {https://doi.org/10.1007/s10878-009-9230-0},
  doi          = {10.1007/s10878-009-9230-0},
}

@inproceedings{WangNFC08,
  author       = {Jianxin Wang and
                  Dan Ning and
                  Qilong Feng and
                  Jianer Chen},
  editor       = {Manindra Agrawal and
                  Ding-Zhu Du and
                  Zhenhua Duan and
                  Angsheng Li},
  title        = {An Improved Parameterized Algorithm for a Generalized Matching Problem},
  booktitle    = {Theory and Applications of Models of Computation, 5th International Conference},
  series       = {Lecture Notes in Computer Science},
  volume       = {4978},
  pages        = {212--222},
  publisher    = {Springer},
  year         = {2008},
  doi          = {10.1007/978-3-540-79228-4\_19},
}

@article{PrietoS06,
  author       = {Elena Prieto-Rodriguez and
                  Christian Sloper},
  title        = {Looking at the stars},
  journal      = {Theor. Comput. Sci.},
  volume       = {351},
  number       = {3},
  pages        = {437--445},
  year         = {2006},
  url          = {https://doi.org/10.1016/j.tcs.2005.10.009},
  doi          = {10.1016/j.tcs.2005.10.009},
}

@article{KanekoKN01,
author = {Kaneko, Atsushi and Kelmans, Alexander and Nishimura, Tsuyoshi},
title = {On packing 3-vertex paths in a graph},
journal = {Journal of Graph Theory},
volume = {36},
number = {4},
pages = {175--197},
doi = {https://doi.org/10.1002/1097-0118(200104)36:4<175::AID-JGT1005>3.0.CO;2-T},
year = {2001}
}

@inproceedings{Karp72,
  author       = {Richard M. Karp},
  editor       = {Raymond E. Miller and
                  James W. Thatcher},
  title        = {Reducibility Among Combinatorial Problems},
  booktitle    = {Proceedings of a symposium on the Complexity of Computer Computations,
                  held March 20-22, 1972, at the {IBM} Thomas J. Watson Research Center,
                  Yorktown Heights, New York, {USA}},
  series       = {The {IBM} Research Symposia Series},
  pages        = {85--103},
  publisher    = {Plenum Press, New York},
  year         = {1972},
  url          = {https://doi.org/10.1007/978-1-4684-2001-2_9},
  doi          = {10.1007/978-1-4684-2001-2_9},
  bibsource    = {dblp computer science bibliography, https://dblp.org}
}

@article{ArnborgCP87,
author = {Arnborg, Stefan and Corneil, Derek G. and Proskurowski, Andrzej},
title = {Complexity of Finding Embeddings in a k-Tree},
journal = {SIAM Journal on Algebraic Discrete Methods},
volume = {8},
number = {2},
pages = {277-284},
year = {1987},
doi = {10.1137/0608024},
}

@article{Poljak74,
author = {Poljak, Svatopluk},
journal = {Commentationes Mathematicae Universitatis Carolinae},
language = {eng},
number = {2},
pages = {307-309},
publisher = {Charles University in Prague, Faculty of Mathematics and Physics},
title = {A note on stable sets and colorings of graphs},
url = {http://eudml.org/doc/16622},
volume = {015},
year = {1974},
}

@techreport{Scheffler94,
    author = {Petra Scheffler},
    title = {Practical linear time algorithm for disjoint paths in graphs with bounded tree-width},
    institution = {FU Berlin, Fachbereich 3 Mathematik},
    year = {1994}
}

@article{GareyJS76,
author = {M.R. Garey and D.S. Johnson and L. Stockmeyer},
title = {Some simplified NP-complete graph problems},
journal = {Theoretical Computer Science},
volume = {1},
number = {3},
pages = {237--267},
year = {1976},
issn = {0304-3975},
doi = {10.1016/0304-3975(76)90059-1},
}
\end{sloppypar}

\end{document}